\documentclass[12pt, draftclsnofoot, onecolumn]{IEEEtran}
\usepackage{amsmath,amssymb}
\usepackage{subfigure}
\usepackage{graphicx,graphics,color,psfrag}
\usepackage{cite,balance}
\usepackage{caption}
\allowdisplaybreaks
\usepackage{algorithm}
\usepackage{algorithmic}
\usepackage{accents}
\usepackage{amsthm}
\usepackage{bm}
\usepackage{url}
\usepackage[english]{babel}
\usepackage{multirow}
\usepackage{enumerate}
\usepackage{cases}
\usepackage{stfloats}
\usepackage{dsfont}
\usepackage{color,soul}
\usepackage{amsfonts}
\usepackage{cite,graphicx,amsmath,amssymb}
\usepackage{subfigure}
\usepackage{fancyhdr}
\usepackage{hhline}
\usepackage{graphicx,graphics}
\usepackage{array,color}
\usepackage{mathtools}
\usepackage{amsmath}
\usepackage{setspace}

\newtheorem{theorem}{\emph{\underline{Theorem}}}

\newtheorem{definition}{\emph{\underline{Definition}}}

\newtheorem{lemma}{\emph{\underline{Lemma}}}
\newtheorem{corollary}{\emph{\underline{Corollary}}}

\newtheorem{example}{\bf Example}
\newtheorem{remark}{\bf \emph{\underline{Remark}}}

\def\({\left(}
\def\){\right)}

\def\b0{{\mathbf{0}}}

\newcommand{\diag}{\mathrm{diag}}

\newcommand{\nn}{\nonumber}

\usepackage{setspace}

\allowdisplaybreaks[4]
\begin{document}

\captionsetup[figure]{name={Fig.}}

\title{\huge 
Multi-Active Multi-Passive (MAMP)-IRS  Aided Wireless Communication: A Multi-Hop \\ Beam Routing Design}
\author{Yunpu~\!Zhang, \!Changsheng~\!You,~\IEEEmembership{\!Member,~\!IEEE},  and \!Beixiong \!Zheng,~\IEEEmembership{\!Member,~\!IEEE}
\thanks{\noindent Part of this work will be presented at the 2022 IEEE Global Communications Conference, Rio de Janeiro, Brazil \cite{zhangSAMP}.

Y. Zhang and C. You are with the Department of Electronic and Electrical Engineering, Southern University of Science and Technology (SUSTech), Shenzhen 518055, China (e-mail: zhangyp2022@mail.sustech.edu.cn, youcs@sustech.edu.cn).

B. Zheng is with the School of Microelectronics, South China University of Technology, Guangzhou 511442, China (e-mail: bxzheng@scut.edu.cn).
}}
\maketitle

\newcommand\blfootnote[1]{%
\begingroup
\renewcommand\thefootnote{}\footnote{#1}%
\addtocounter{footnote}{-1}%
\endgroup
}
\vspace{-1.5cm}
\begin{abstract} 
Prior studies on intelligent reflecting surface (IRS) have mostly considered wireless communication systems aided by a \emph{single passive} IRS, which, however, has limited control over wireless propagation environment and suffers severe product-distance path-loss. To address these issues, we propose in this paper a new \emph{multi-active multi-passive (MAMP)}-IRS aided wireless communication system, where a number of active and passive IRSs are deployed to assist the communication between a base station (BS) and a remote user in complex environment, by establishing a multi-hop reflection path across active and passive IRSs. In particular, the active IRSs enable to opportunistically amplify the reflected signal along the multi-reflection link, thus effectively compensating for the severe product-distance path-loss. For the new MAMP-IRS aided system, an optimization problem is formulated to maximize the achievable rate of a typical user by designing the active-and-passive IRS  routing path as well as the joint beamforming of the BS and selected active/passive IRSs. To draw useful insights into the optimal design, we first consider a special case of the \emph{single-active multi-passive (SAMP)}-IRS aided system. For this case, we propose an efficient algorithm to obtain its optimal solution by first optimizing the joint beamforming given any SAMP-IRS routing path, and then optimizing the routing path by using a new path decomposition method and graph theory. 
Moreover, we show that the active IRS should be selected to establish the beam routing path when its amplification power and/or number of active reflecting elements are sufficiently large. Next, for the general MAMP-IRS aided system, we show that its challenging beam routing optimization problem can be efficiently solved by a new two-phase approach. Its key idea is to first optimize the inner passive-IRS beam routing between each two active IRSs for effective channel power gain maximization, followed by an outer active-IRS beam routing optimization for rate maximization. Last, numerical results are provided to validate our analytical results and demonstrate the effectiveness of the proposed MAMP-IRS beam routing scheme as compared to various benchmark schemes.
\end{abstract}

\vspace{-0.7cm}
\IEEEpeerreviewmaketitle

\vspace{-0.7cm}
\section{Introduction}
\vspace{-0.2cm}
Intelligent reflecting surface (IRS) has emerged as a promising technology to smartly reconfigure the wireless propagation environment by dynamically tuning its reflecting elements \cite{wu2021intelligent,di2020smart,zheng2022survey}.  Specifically, IRS is a planar metasurface composed of a large number of reflecting elements, each of which is able to independently control the amplitude and/or phase of the incident signal, thereby collaboratively reconfiguring the wireless channels. Moreover, different from the conventional active transmitters and relays, IRS 
passively reflects the signals without installing costly radio-frequency (RF) chains for signal transmission/reception, hence significantly reducing the hardware and energy cost. These appealing features thus have motivated substantial research recently to incorporate IRS into traditional wireless systems for improving the communication performance (see, e.g., \cite{liu2021reconfigurable,9801736,huang2019reconfigurable,9745477,9483903,9198125,wang2020channel,zhou2020framework,9133130}).

In the existing works on IRS, most of them have considered the basic wireless communication system aided by a \emph{single passive} IRS \cite{pan2020multicell,wu2019intelligent,you2020channel,zhang2021intelligent,yang2021energy,hu2020location,9724202}. However, employing only a single passive IRS usually has limited control over the wireless channels, which may not be able to unlock the full potentials of IRS. For example, in complex environment (e.g., corridors in an indoor setting), there may not exist a blockage-free reflection link between the base station (BS) and a remote user via a single IRS only \cite{mei2022intelligent,9690635,9410457}. To address this issue, new research efforts have been recently devoted to designing more efficient multi-IRS aided systems by deploying two or more IRSs in the network. This thus enables multiple IRSs to collaboratively establish double- or multi-hop signal reflection paths from the BS to the users so as to bypass the obstacles in between (see, e.g., \cite{han2020cooperative,you2020wireless,zheng2021double,9373363,mei2020cooperative,mei2021multi}). In particular, the authors in \cite{han2020cooperative} showed that under the line-of-sight (LoS) channel model, the double-IRS aided single-user system can achieve a higher scaling order of the passive beamforming gain than the conventional single-IRS counterpart, i.e., $\mathcal{O}(M^4)$ versus $\mathcal{O}(M^2)$ \cite{wu2019intelligent} with $M$ denoting the number of passive reflecting elements. Moreover, the authors in \cite{you2020wireless} showed that even with the channel estimation error and training overhead taken into account,  two cooperative IRSs can still achieve significant rate improvement over the conventional single-IRS system. In addition, it was shown in \cite{zheng2021double} that deploying two cooperative IRSs
near the BS and users achieves superior performance than the single-IRS system, in terms of the maximum
signal-to-noise ratio (SNR) and multi-user effective channel rank. Furthermore, for the more complex environment (e.g., smart factories) wherein deploying two IRSs may not be able to provide virtual LoS communication links,  a multi-IRS aided system  was proposed in \cite{mei2020cooperative}, where the remote user is served by the BS via a multi-reflection link with the signal consecutively reflected by a set of selected IRSs. To balance the tradeoff between the multiplicative passive beamforming gain and product-distance path-loss, efficient beam routing algorithms were proposed to maximize the effective channel power gain. Besides, the more general multi-IRS aided multi-user communication system was further studied in \cite{mei2021multi}, where  an efficient recursive algorithm was proposed to maximize the
minimum received signal power among all users. However, despite the prominent multiplicative passive beamforming gain, the double/multi-passive-IRS aided systems still suffer severe \emph{product-distance} path-loss arising from the double/multi-hop signal reflection.

To deal with the above issue, a new type of IRS architecture, called \emph{active} IRS, has been recently proposed \cite{zhang2021active,long2021active,9568854}. Specifically, equipped with low-power reflection-type amplifiers (e.g., tunnel diode and negative impedance converter), active IRS is capable of reflecting incident signals with \emph{power amplification} in a full-duplex mode, thus effectively compensating for the product-distance path-loss at a modest energy and hardware cost \cite{zhang2021active,long2021active,9568854,you2021wireless,zeng2022throughput,zhi2022active,kang2022active}. Moreover, active IRS generally outperforms the conventional amplify-and-forward (AF) relay, since the latter usually operates in a half-duplex mode and thus suffers lower spectrum efficiency. These promising features thus have motivated upsurging research interests recently on active-IRS aided systems. For example, it was shown in \cite{zhang2021active} and \cite{long2021active} that given the same IRS location, the active-IRS aided system tends to achieve better performance than the passive-IRS counterpart thanks to the power amplification gain. Besides, the authors in \cite{you2021wireless} compared the optimal IRS placement strategy  and the achievable rate for the active- and passive-IRS aided wireless systems, and showed that active IRS attains a higher rate than passive IRS when its amplification power is sufficiently large and/or number of active reflecting
elements is relatively small.
Moreover, it was revealed in \cite{zhi2022active} that under the total system energy consumption constraint, active IRS is superior to passive IRS when its number of reflecting elements is small and the total power budget is  sufficiently large. Furthermore, to reap both the active-IRS power amplification gain and passive-IRS beamforming gain, the authors in \cite{kang2022active} proposed a new hybrid active-passive IRS architecture consisting of both active and passive reflecting elements, based on which they devised an optimal active/passive elements allocation strategy for rate maximization. However, this hybrid IRS architecture may increase both the hardware and design complexity, which thus limits its large-scale deployment in practice. Instead, deploying multiple active and passive IRSs in a distributed manner is expected to significantly improve the communication performance in complex environment, which, however, has not been studied in the existing literature. In particular, it remains unknown how to leverage active and passive IRSs for constructing effective multi-hop beam routing paths and whether the newly added active IRSs can bring significant performance gain to the multi-IRS aided wireless communication system. 
\begin{figure}[t]
    \centering
    \includegraphics[width=10cm]{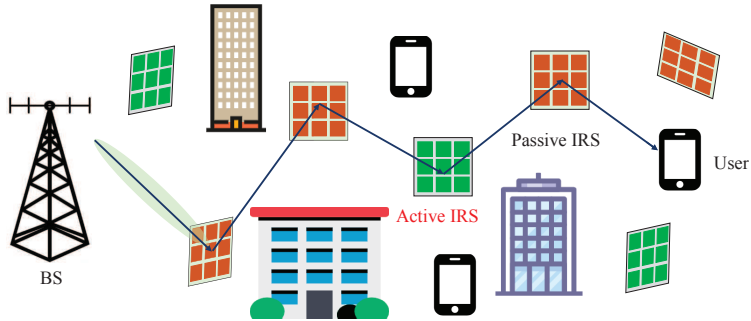}
    \caption{MAMP-IRS aided wireless communication systems.}\label{fig:system}
    \vspace{-22pt}
\end{figure}

    Motivated by the above, we consider in this paper a new \emph{multi-active multi-passive (MAMP)}-IRS aided wireless communication system as shown in Fig.~1, where a number of active and passive IRSs are deployed to assist the downlink communication via establishing a cascaded LoS link for each user. Compared with the conventional multi-passive-IRS aided wireless system \cite{mei2020cooperative}, the new MAMP-IRS system enables to amplify the reflected signal \emph{opportunistically} along the multi-reflection link with the aid of active IRSs, thus effectively compensating for the severe multi-reflection path-loss at the expense of a modest higher energy/hardware cost. However, the new MAMP-IRS beam routing optimization  problem is also much more challenging, since it needs to strike a new trade-off between maximizing the power amplification gain of  selected  active IRSs and minimizing the \emph{accumulated}  amplification noise at the user. To this end, we formulate an optimization problem for  maximizing the achievable rate of a typical user by jointly designing the MAMP-IRS beam routing as well as the beamforming of the BS and selected active/passive IRSs. The main contributions of this paper are summarized as follows.
\begin{itemize}
    \item First, to obtain useful insights, we consider a special system setup with one active IRS and multiple passive IRSs deployed, termed as \emph{single-active multi-passive (SAMP)}-IRS aided wireless communication. The corresponding SAMP-IRS beam-routing optimization problem is shown to be non-convex, which is in general difficult to be solved optimally. To address this issue,
    we first obtain the optimal beamforming of the BS and selected (active/passive) IRSs for any feasible SAMP-IRS routing path. Then, we show that the SAMP-IRS routing optimization problem can be equivalently decomposed into two separate sub-path routing designs, corresponding to the multi-reflection path from the BS to the active IRS and that from the active IRS to the user, respectively; thus significantly reducing the routing design complexity. Subsequently, an efficient graph-optimization algorithm is proposed to obtain the optimal SAMP-IRS routing path by using graph theory. In particular, we show that the active IRS should be selected to establish the beam routing path when its amplification power and/or number of active reflecting elements are sufficiently large.
    \item Second, we consider the general case of the MAMP-IRS aided system. The corresponding rate maximization problem is much more challenging than that of the SAMP-IRS case due to the (possible) cascaded power amplification and accumulated amplification noise arising from selected active IRSs. To tackle this difficulty, we first show that the MAMP-IRS beam routing optimization problem can be  efficiently solved by a
    new \emph{two-phase} approach, which first optimizes the 
    \emph{inner passive-IRS beam routing} between each two active IRSs for effective channel power gain maximization, and then optimizes the  \emph{outer active-IRS beam routing} for maximizing the achievable rate. In particular, to solve the challenging (outer) active-IRS routing problem, we approximate  the receive SNR by its lower bound which is asymptotically accurate in the high-SNR regime. As such,  
    the problem is transformed into a more tractable form and then solved by an efficient graph-optimization algorithm.
    \item Last, we provide numerical results to corroborate our
    theoretical findings on the new MAMP-IRS aided
    system as well as validate the effectiveness of the proposed MAMP-IRS beam routing scheme. It is shown that the MAMP-IRS aided system  achieves significantly higher rate than the multi-passive-IRS aided system under various system setups. It is also observed that our proposed MAMP-IRS beam routing scheme greatly outperforms the other benchmark schemes.
\end{itemize}

The remainder of this paper is organized as follows. Sections~\ref{System Model and Problem Formulation} and \ref{Probelm Formu} introduce the system model and problem formulation for the MAMP-IRS aided system, respectively. In Sections \ref{SAMP} and \ref{MAMP}, we propose efficient algorithms to design the multi-hop beam routing in the SAMP-IRS and MAMP-IRS aided systems, respectively. Numerical results are presented in Section~\ref{Results} to evaluate the performance of the proposed MAMP-IRS beam routing design, followed by the conclusions given in Section~\ref{CL}.

\emph{Notations:} 
The superscript $(\cdot)^{H}$ denotes the Hermitian transpose. For a complex vector $\boldsymbol{x}$, $\|\boldsymbol{x}\|$ denotes the $\ell_{2}$-norm, and $\angle(\cdot)$  denotes the phase of each element in $\boldsymbol{x}$. $|\cdot|$ denotes the absolute value for a real number and the cardinality for a set. $\otimes$ denotes the Kronecker product. $\mathcal{O}(\cdot)$ denotes the standard big-O notation. For matrices, $\diag(\cdot)$ denotes a square diagonal matrix with the elements in $(\cdot)$ on its main diagonal. $\mathbf{I}_{N}$ denotes an identity matrix  with size $N\times N$, and $\mathbf{0}_{N}$ denotes an $N\times N$ all-zero matrix. $\mathbb{C}^{N \times M}$ denotes the space of $N \times M$ of complex-valued matrices. The main symbols used in this paper are summarized in Table~\ref{Table1}. 
\begin{table*}[t]
\renewcommand\arraystretch{1.4}
\centering
{\color{black}{
\caption{List of main symbols and their physical meanings.}\label{Table1}
\scalebox{0.75}{
\begin{tabular}{| c | l ||c| l | }
\hline
\multirow{1}{*}{$J_{a}$}&\multirow{1}{*}{Number of active IRSs}&\multirow{1}{*}{$J_{p}$}&\multirow{1}{*}{Number of passive IRSs}\\
\hline
\multirow{1}{*}{$N$}&\multirow{1}{*}{Number of elements at each active IRS}&\multirow{1}{*}{$M$}&\multirow{1}{*}{Number of elements at each passive IRS}\\
\hline
\multirow{1}{*}{$T$}&\multirow{1}{*}{Number of BS antennas}&\multirow{1}{*}{$J$}&\multirow{1}{*}{Total number of IRSs}\\
\hline
\multirow{1}{*}{$\mathcal{J}$}&\multirow{1}{*}{Set of IRSs}&\multirow{1}{*}{$\mathcal{J}_{a}$}&\multirow{1}{*}{Set of active IRSs}\\
\hline
\multirow{1}{*}{$\mathcal{J}_{p}$}&\multirow{1}{*}{Set of passive IRSs}&\multirow{1}{*}{$\boldsymbol{H}_{0, j}$}&\multirow{1}{*}{Channel from the BS to IRS $j$}\\
\hline
\multirow{1}{*}{$\boldsymbol{H}_{i, j}$}&\multirow{1}{*}{Channel from IRS $i$ to IRS $j$}& \multirow{1}{*}{$\boldsymbol{h}^H_{i, J+1}$}&\multirow{1}{*}{Channel from IRS $i$ to user}\\
\hline
\multirow{1}{*}{$\boldsymbol{\Psi}_{j}$}&Reflection coefficient matrix of passive IRS $j$ &\multirow{1}{*}{$\boldsymbol{\Phi}_{\ell}$}&\multirow{1}{*}{Reflection coefficient matrix of active IRS $\ell$ }\\
\hline
\multirow{1}{*}{$\eta_{\ell}$}&Amplification factor of active IRS $\ell$ &\multirow{1}{*}{$\sigma^2_{\rm F}$}&\multirow{1}{*}{Amplification noise power}\\
\hline
\multirow{1}{*}{$\Omega$}&\multirow{1}{*}{MAMP-IRS routing path}&\multirow{1}{*}{$K$}& Total number of selected IRSs\\
\hline
\multirow{1}{*}{$K_{a}$}&\multirow{1}{*}{Number of selected active IRSs}&\multirow{1}{*}{$K_{p}$}& Number of selected passive IRSs\\
\hline
\multirow{1}{*}{$\mu(k_{a})$}&\multirow{1}{*}{Routing-index of the $k_{a}$-th selected active IRS}&\multirow{1}{*}{$\mathcal{K}_{a}$}& Node-indices of the selected active IRSs\\
\hline
\multirow{1}{*}{$\boldsymbol{G}_{0,s_{\mu(1)}}$}&\multirow{1}{*}{Channel from the BS to active IRS $s_{\mu(1)}$ }&\multirow{1}{*}{$\boldsymbol{G}_{s_{\mu(k_{a})},s_{\mu(k_{a}+1)}}$}& Channel from active IRS $s_{\mu(k_{a})}$ to $s_{\mu(k_{a}+1)}$ \\
\hline
\multirow{1}{*}{$\boldsymbol{g}^{H}_{s_{\mu(K_{a})},J+1}$}&\multirow{1}{*}{Channel from active IRS $s_{\mu(K_{a})}$ to the user}&\multirow{1}{*}{$\boldsymbol{g}^{H}_{0,J+1}$}& Channel from the BS to the user \\
\hline
\multirow{1}{*}{$\boldsymbol{G}_{0,s_{\mu(k_{a})}}$}&\multirow{1}{*}{Channel from the BS to active IRS $s_{\mu(k_{a})}$ }&\multirow{1}{*}{$\boldsymbol{g}^{H}_{s_{\mu(k_{a})},J+1}$}& Channel from active IRS $s_{\mu(k_{a})}$ to the user \\
\hline
\multirow{1}{*}{$\Omega_{\rm act}$}&\multirow{1}{*}{Active-IRS routing path}&\multirow{1}{*}{$\gamma_{i}$}& Per-hop SNR between nodes ($i-1$) and $i$\\
\hline
\end{tabular}}}}
\vspace{-20pt}
\end{table*}
\vspace{-0.3cm}
\section{System Model}\label{System Model and Problem Formulation}
\vspace{-0.1cm}
As shown in Fig. 1, we consider a MAMP-IRS aided wireless  communication system, where $J_{a}$ active IRSs (each equipped with $N$ elements) and $J_{p}$ passive IRSs (each equipped with $M$ elements) are deployed to assist the downlink communication from a $T$-antenna BS to multiple remote users. In particular, we consider the time division multiple access (TDMA) to separate the communications for different users, and thus focus on the design for a typical user in the paper, without loss of generality.\footnote{The obtained results can be extended to the case of simultaneous multi-user communications by designing the multi-beam routing for different users as in \cite{mei2021multi}.} We assume that the direct BS$\to$user channel is severely blocked by scattered obstacles; thus the BS can communicate with the user via a multi-reflection path only that is formed by a set of selected active/passive IRSs.\footnote{The obtained results can be easily extended to the case where a direct link exists between the BS and (typical) user, by aligning the multi-reflection link with the direct link in phase \cite{mei2021performance}.} Let $\bar{\mathcal{J}}\triangleq\{0, \mathcal{J}, J+1\}$ denote the set of all nodes, where nodes $0$ and $J + 1$ refer to the BS and user, respectively, and  $\mathcal{J}\triangleq\{1,2, \cdots, J\}$ stands for the set of all IRSs. In addition, the sets of $J_{a}$ active IRSs and $J_{p}$ passive IRSs are denoted by $\mathcal{J}_{a} $, and $\mathcal{J}_{p}$ (i.e., $\mathcal{J}_{a} \cup \mathcal{J}_{p} = \mathcal{J}$), respectively. Besides, the number of elements/antennas of each node is denoted by $U_j, \forall j\in \bar{\mathcal{J}}$, e.g., $U_0=T$, $U_{j}=N, \forall j \in \mathcal{J}_{a}$ and $U_{j}=M, \forall j \in \mathcal{J}_{p}$.
\vspace{-0.3cm}
\subsection{Channel Model}
\vspace{-0.1cm}
Let $\boldsymbol{H}_{0, j}$, $j \in {\mathcal{J}}$  denote the channel from the BS to the active/passive IRS $j$; $\boldsymbol{H}_{i, j}$, $i, j \in {\mathcal{J}}, i\neq j$ represent the channel from IRS $i$ to IRS $j$; and $\boldsymbol{h}_{i, J+1}^{H}$, $i \in \mathcal{J}$ denote the channel from the IRS to the user, where the dimension of each channel depends on the number of antennas/elements of the associated nodes. The above channels are modeled as follows.

Similar to \cite{mei2020cooperative}, we consider the dominant LoS channel components for the IRS reflection design, while treating all non-LoS (NLoS) components as part of environment scattering, which has been shown to have a marginal effect on the user performance\cite{mei2021multi}.\footnote{Numerical results will be provided in Section~\ref{Results} to show the effect of Rician factor on the rate performance.} Moreover, the information of the LoS availability for all links is assumed to be known \emph{a priori}, which can be practically acquired by using e.g., efficient offline/online IRS beam training methods \cite{you2020fast,mei2022intelligent}. For each communication link involved, we denote by $a_{i,j}$, $i,j\in \bar{\mathcal{J}}$ the channel state indicator, where $a_{i,j}=1$ indicates that there exists an LoS link between nodes $i$ and $j$, and $a_{i,j}=0$ otherwise. Then, if $a_{i,j}=1$ for $i,j \in {\mathcal{J}}, i\neq j$, the inter-IRS channel, $\boldsymbol{H}_{i, j}$, can be modeled as
\vspace{-0.3cm}
\begin{equation}
\boldsymbol{H}_{i, j}=h_{i,j}\boldsymbol{a}_{\rm r}\left(\vartheta^{\rm r}_{i,j}, \theta^{\rm r}_{i,j}, U_j\right) \boldsymbol{a}_{\rm t}^H \left(\vartheta^{\rm t}_{i,j}, \theta^{\rm t}_{i,j}, U_i\right),\nn
\vspace{-0.3cm}
\end{equation}
where $h_{i,j}=\frac{\sqrt{\beta}}{d_{i, j}} e^{-\frac{\jmath 2 \pi d_{i, j}}{\lambda}}$ denotes the complex channel gain for the link with $d_{i,j}$ being the link distance, $\beta$ denoting the reference channel gain at a distance of $1$ meter (m), and $\lambda$ denoting the carrier wavelength.  Moreover, $\vartheta^{\rm r}_{i,j}$ (or $\theta^{\rm r}_{i,j})$  denotes the azimuth (or elevation) angle-of-arrival (AoA) at node $j$ from node $i$; $\vartheta^{\rm t}_{i,j}$ (or $\theta^{\rm t}_{i,j})$ denotes the azimuth (or elevation) angle-of-departure (AoD) from node $i$ to node $j$; and $\boldsymbol{a}_{\rm r}$ and $\boldsymbol{a}_{\rm t}$ denote respectively the receive and transmit steering vectors. Specifically, based on the uniform rectangular array (URA) model for the IRS, $\boldsymbol{a}_{\rm r}$ can be modeled as 
$\boldsymbol{a}_{\rm r}\left(\vartheta^{\rm r}_{i,j}, \theta^{\rm r}_{i,j}, U_j\right)= \boldsymbol{u}(\frac{2 d_{\rm I}}{\lambda}\sin\theta^{\rm r}_{i,j} \cos\vartheta^{\rm r}_{i,j}, U_j^{(1)})\otimes \boldsymbol{u}(\frac{2 d_{\rm I}}{\lambda}\cos\theta^{\rm r}_{i,j}, U_j^{(2)})$, where $U_j^{(1)}$ and $U_j^{(2)}$ denote respectively the numbers of horizontal and vertical elements of node $j$, and the function $\boldsymbol{u}$ is defined as
$\boldsymbol{u}(\zeta, U)\triangleq [1, e^{-\jmath \pi \zeta}, \dots, e^{-\jmath (U-1) \pi \zeta}]$.
Following the above, the transmit steering vector, $\boldsymbol{a}_{\rm t}$, can be similarly defined. Besides, the BS$\to$IRS channel, $\boldsymbol{H}_{0, j}$, can be  modeled as
\vspace{-0.2cm}
\begin{equation}
\boldsymbol{H}_{0, j}=h_{0,j}\boldsymbol{a}_{\rm r}\left(\vartheta^{\rm r}_{0,j}, \theta^{\rm r}_{0,j}, U_j\right) \boldsymbol{a}_{\rm t}^H \left(\vartheta^{\rm t}_{0,j}, U_0\right),~~ j \in \mathcal{J}, \nn
\vspace{-0.3cm}
\end{equation}
where  $\boldsymbol{a}_{\rm t}^H \left(\vartheta^{\rm t}_{0,j}, U_0\right)=\boldsymbol{u}^H(\frac{2 d_{\rm I}}{\lambda}\cos\theta^{\rm t}_{0,j}, U_0)$ is the transmit steering vector of the BS based on the uniform linear array (ULA) model. The IRS$\to$user channel can be modeled as
\vspace{-0.3cm}
\begin{equation}
\boldsymbol{h}^H_{i, J+1}=h_{i,J+1}\boldsymbol{a}_{\rm t}^H \left(\vartheta^{\rm t}_{i,J+1}, \theta^{\rm t}_{i,J+1}, U_i\right),~~ i \in \mathcal{J}.\nn
\vspace{-0.3cm}
\end{equation}
\vspace{-1.3cm}
\subsection{Signal Reflection Model of Active/Passive IRS}
\vspace{-0.2cm}
First, for each passive IRS $j \in \mathcal{J}_{p}$, we denote by $\boldsymbol{\Psi}_{j}=\operatorname{diag}\left(e^{\jmath \phi_{j, 1}}, e^{\jmath \phi_{j, 2}}, \cdots, e^{\jmath \phi_{j, M}}\right)$ its passive reflection matrix, where we set the reflection amplitude as one (i.e., its maximum value) to maximize the signal reflection power, and $\phi_{j, m}\in[0, 2\pi)$ represents the phase-shift at its passive element $m \in \mathcal{M} \triangleq\{1, \cdots, M\}$. Next, for each active IRS $\ell \in \mathcal{J}_{a}$, we denote its reflection matrix by  $\boldsymbol{\Psi}_{\ell}=\operatorname{diag}\left(\eta_{\ell, 1}e^{\jmath \phi_{\ell, 1}}, \eta_{\ell, 2}e^{\jmath \phi_{\ell, 2}}, \cdots, \eta_{\ell, N}e^{\jmath \phi_{\ell, N}}\right)$, where $\eta_{\ell, n}$ and $\phi_{\ell, n}$ represent the reflection amplitude and phase-shift at each active element $n \in \mathcal{N} \triangleq  \{1, \cdots, N\} $, respectively. 
Moreover, based on the LoS channel model, it can be easily shown that without loss of optimality, all active reflecting elements for each active IRS $\ell$ should adopt the same amplification factor, i.e., $\eta_{\ell, n}=\eta_{\ell}$, $\forall n \in \mathcal{N}, \ell\in \mathcal{J}_{a}$. Accordingly, $\boldsymbol{\Psi}_{\ell}$ can be equivalently expressed as $\boldsymbol{\Psi}_{\ell}=\eta_{\ell}\boldsymbol{\Phi}_{\ell}$, where $\boldsymbol{\Phi}_{\ell} \triangleq \operatorname{diag}\left(e^{\jmath \phi_{\ell, 1}}, e^{\jmath \phi_{\ell, 2}}, \cdots, e^{\jmath \phi_{\ell, N}}\right)$. It is worth noting that each active IRS incurs non-negligible thermal noise at its reflecting elements.
For each active IRS, we denote its amplification noise by $\boldsymbol{n}_{\mathrm{F}} \in \mathbb{C}^{N \times 1}$,
which is assumed to follow the independent circularly symmetric complex Gaussian (CSCG) distribution, i.e., $\boldsymbol{n}_{\mathrm{F}} \sim \mathcal{C} \mathcal{N}\left(\mathbf{0}_{N}, \sigma_{\mathrm{F}}^{2} \mathbf{I}_{N}\right)$ with $\sigma_{\mathrm{F}}^{2}$ denoting the amplification noise power.
\vspace{-1.1cm}
\section{Problem Formulation}\label{Probelm Formu}
\vspace{-0.2cm}
In this section, we formulate an optimization problem for maximizing the user's achievable rate by designing the MAMP-IRS beam routing. 

It is noted that the signal models for the MAMP-IRS beam
routing with and without the active IRSs involved are in different forms. Thus, we divide the MAMP-IRS beam routing design into two cases and formulate their corresponding optimization problems. For ease of exposition, we first define the feasible routing path as follows.
\vspace{-6pt}
\begin{definition}[Feasible routing path]\label{De:Feasible}
\emph{Let $\Omega=\{s_{1}, s_{2}, \ldots, s_{ K}\}$ denote a multi-reflection routing path, where $K$ represents the number of selected active/passive IRSs and $s_{k} \in \mathcal{J}$ with $k \in \mathcal{K} \triangleq  \{1, \cdots, K\}$ denotes the node-index of the $k$-th selected IRS. Then, $\Omega$ is feasible if the following conditions are satisfied
\vspace{-0.6cm}
\begin{align}
&s_{k} \in \mathcal{J}, s_{k} \neq s_{k^{\prime}},~~\forall k,k^{\prime}\in\mathcal{K}, k \neq k^{\prime},\label{Eq:Feasible1}\\
&a_{s_{k},s_{k+1}}=1,~~\forall k\in\mathcal{K}, k \neq K,\label{Eq:Feasible2}\\
&a_{0,s_{1}}=a_{s_{K},J+1}=1.\label{Eq:Feasible3}
\end{align}}
\vspace{-1.5cm}
\end{definition}
\noindent Note that the condition in \eqref{Eq:Feasible1} ensures that each IRS in $\mathcal{J}$ is selected at most once; the condition in \eqref{Eq:Feasible2} enforces that there exists an LoS link between any two adjacent nodes along the routing path $\Omega$; and the last condition in \eqref{Eq:Feasible3} guarantees that the multi-reflection routing path $\Omega$ starts from the BS and terminates at the user. Based on the above definition, we then introduce the signal models for the two cases without and with active IRSs involved.

\vspace{-0.6cm}
\subsection{Passive-IRS Multi-reflection Signal Model}
\vspace{-0.2cm}
First, consider the case where the active IRSs are assumed not selected when establishing the multi-reflection path. Let $\tilde{\Omega}=\left\{s_{1}, s_{2}, \ldots, s_{\tilde K}\right\}$ define a feasible routing path from the BS to the user via passive IRSs only, where $\tilde{K}$ is the number of selected passive IRSs and $s_{\tilde{k}} \in \mathcal{J}_{p}$ with $\tilde{k} \in \tilde{\mathcal{K}} \triangleq  \{1, \cdots, \tilde{K}\}$ denotes the node-index of the $\tilde{k}$-th selected passive IRS. Then the BS$\to$user equivalent channel is given by
\vspace{-0.3cm}
\begin{equation}\label{Eq:D_Passive}
\tilde{\boldsymbol {g}}_{\rm BU}^H=\boldsymbol{h}_{s_{\tilde K}, J+1}^{H} \boldsymbol{\Psi}_{s_{\tilde K}} \prod_{{\tilde k} \in \{\mathcal{\tilde K} \setminus \tilde K\}}\left(\boldsymbol{H}_{s_{{\tilde k}}, s_{{\tilde k}+1}} \boldsymbol{\Psi}_{s_{{\tilde k}}}\right) \boldsymbol{H}_{0, s_{1}}.
\vspace{-0.3cm}
\end{equation}
As such, the received signal at the user via the multi-reflection path $\tilde{\Omega}$ with passive IRSs only is given as
\vspace{-0.2cm}
\begin{equation}
\tilde{y}_{\mathrm{pas}}=\tilde{\boldsymbol {g}}_{\rm BU}^H   \boldsymbol{w}_{\rm B}x+n_0,
\vspace{-10pt}
\end{equation}
where $x$ denotes the transmitted signal at the BS with power $P_{\mathrm{B}}$, $\boldsymbol{w}_{\rm B} \in \mathbb{C}^{T\times 1}$ denotes the normalized beamforming vector of the BS with  $||\boldsymbol{w}_{\rm B} ||^{2}=1$, $n_0$ denotes the received Gaussian noise at the user with power $\sigma^{2}$. The corresponding achievable rate of the user in bits/second/Herz (bps/Hz) is thus given by
$\tilde{R}_{\mathrm{pas}}=\log_{2}\left(1+\tilde{\gamma}_{\mathrm{pas}}\right)$, where the receive SNR is 
\vspace{-0.3cm}
\begin{equation}
\tilde{\gamma}_{\mathrm{pas}}={P_{\mathrm{B}} |\tilde{\boldsymbol{g}}_{\mathrm{BU}}^{H}  \boldsymbol{w}_{\rm B}|^{2}}/{\sigma^{2}}.
\vspace{-0.5cm}
\end{equation}
Based on the above, the optimization problem for the case with passive-IRS involved in the beam routing only can be formulated as
\vspace{-0.5cm}
\begin{subequations}
\begin{align}
({\bf P1}):~~\max_{\substack{\boldsymbol{w}_{\rm B}, \{\boldsymbol{\Psi}_{s_{\tilde{k}}}\}, \tilde{\Omega}} }  ~&\log_{2} \left(1+\frac{ P_{\mathrm{B}} |\tilde{\boldsymbol{g}}_{\mathrm{BU}}^{H}  \boldsymbol{w}_{\rm B}|^2 }{\sigma^{2}}\right)
\nn\\
\text{s.t.}~~~~
& \eqref{Eq:Feasible1}-\eqref{Eq:Feasible3},\nn\\
& |[\boldsymbol{\Psi}_{s_{\tilde{k}}}]_{m,m}|=1,~~\forall m \in\mathcal{M},\; \tilde{k} \in \tilde{K}. \label{C:P_unit}
\end{align}
\end{subequations}
It can be verified that problem (P$1$) is a non-convex optimization problem, due to the integer variables $\{s_{\tilde{k}}\}$, unit-modulus constraint in~\eqref{C:P_unit}, as well as the intrinsic coupling between $\tilde{\boldsymbol{g}}_{\mathrm{BU}}^{H}$ and $\boldsymbol{w}_{\rm B}$ in the objective function, which is thus difficult to solve. Generally speaking, there is no standard algorithm for solving this non-convex problem optimally. Nevertheless, it has been shown in \cite{mei2020cooperative} that the optimal solution to problem (P$1$) can be obtained by first designing the joint beamforming of the BS and passive IRSs, and then solving the passive-IRS routing problem by using graph theory. The details are omitted for brevity.
\vspace{-0.5cm}
\subsection{MAMP-IRS Multi-reflection Signal Model}
\vspace{-0.2cm}
Next, we consider the case where active IRSs are involved in the multi-reflection path. Specifically, let
$\Omega=\left\{s_{1}, s_{2}, \ldots, s_{ K}\right\}$ denote a feasible routing path from the BS to the user, which includes $K_{a}$ active IRSs and $K_{p}$ passive IRSs with $K_{a}+K_{p}=K$. Moreover, for the $K_{a}$ selected active IRSs, we denote by $\mu(k_{a})$ the routing-index of the $k_{a}$-th selected active IRS and hence node $s_{\mu(k_{a})}$ is an active IRS, i.e., $s_{\mu(k_{a})} \in \mathcal{J}_{a}$. Then, the node-indices of the selected active IRSs are given by $\mathcal{K}_{a} \triangleq \{\mu(1), \cdots, \mu(k_{a}), \cdots, \mu(K_{a})\}$, while  $\{s_{k_{p}} | s_{k_{p}}\in \mathcal{J}_{p}$, $k_{p} \in \mathcal{K} \setminus \mathcal{K}_{a}\}$ denote the passive IRSs. Let
$\boldsymbol{G}_{0, s_{\mu(1)}}$,  $\boldsymbol{G}_{s_{\mu(k_{a})}, s_{\mu(k_{a}+1)}}$, and $\boldsymbol{g}^{H}_{s_{\mu(K_{a})}, J+1}$ denote the effective channels from the BS to the  first selected active IRS, from the $k_{a}$-th selected active IRS to the ($k_{a}+1$)-th selected active IRS, and from the $K_{a}$-th selected active IRS to the user, respectively, which are given by 
\vspace{-0.3cm}
\begin{equation}\label{Eqs:inter-PIRS}
\begin{cases}\boldsymbol{G}_{0, s_{\mu(1)}} \triangleq \left(\prod_{k=1}^{\mu(1)-1}\boldsymbol{H}_{s_{k},s_{k+1}} \boldsymbol{\Psi}_{s_{k}}\right) \boldsymbol{H}_{0,s_{1}},\\ \boldsymbol{G}_{s_{\mu(k_{a})}, s_{\mu(k_{a}+1)}}\triangleq  \left(\prod_{k=\mu(k_{a})+1}^{\mu(k_{a}+1)-1}\boldsymbol{H}_{s_{k},s_{k+1}}\boldsymbol{\Psi}_{s_{k}}\right)\boldsymbol{H}_{s_{\mu(k_{a})},s_{\mu(k_{a})+1}},\\ \boldsymbol{g}^{H}_{s_{\mu(K_{a})}, J+1}\triangleq \boldsymbol{h}_{s_{K}, J+1}^{H}\boldsymbol{\Psi}_{s_{K}}
\left(\prod_{k=\mu(K_{a})+1}^{K-1}\boldsymbol{H}_{s_{k},s_{k+1}}\boldsymbol{\Psi}_{s_{k}}\right)\boldsymbol{H}_{s_{\mu(K_{a})},s_{\mu(K_{a})+1}}.\end{cases}
\end{equation}
As such, the BS$\to$user equivalent channel over the routing path $\Omega$ can be modeled as
\begin{equation}\label{Eq:BU_Channel}
\boldsymbol{g}^H_{ 0,J+1}=\boldsymbol{g}^{H}_{s_{\mu(K_{a})}, J+1}\eta_{s_{\mu(K_{a})}} \boldsymbol{\Phi}_{s_{\mu(K_{a})}} 
    \left(\prod_{k_{\mathrm{a}}=1}^{K_{\mathrm{a}}-1}  \boldsymbol{G}_{s_{\mu(k_{a})}, s_{\mu(k_{a}+1)}}\eta_{s_{\mu(k_{a})}} \boldsymbol{\Phi}_{s_{\mu(k_{a})}}\right) \boldsymbol{G}_{0, s_{\mu(1)}}.
    \vspace{-6pt}
\end{equation}
It is worth mentioning that each selected active IRS amplifies both the incident signal and non-negligible amplification noise at each active reflecting element. Let $\boldsymbol{G}_{0,s_{\mu(k_{a})}}$ denote the effective channel from the BS to the selected active IRS $s_{\mu(k_{a})}$, defined as
\vspace{-0.3cm}
\begin{equation}\label{Eq:BS_ac}
\boldsymbol{G}_{0,s_{\mu(k_{a})}} \triangleq\begin{cases} \boldsymbol{G}_{0, s_{\mu(1)}} & \text { if } k_{a}=1\\ 
\left(\prod_{\ell=1}^{k_{a}-1}  \boldsymbol{G}_{s_{\mu(\ell)}, s_{\mu(\ell+1)}}\eta_{s_{\mu(\ell)}} \boldsymbol{\Phi}_{s_{\mu(\ell)}}\right) \boldsymbol{G}_{0, s_{\mu(1)}}& \text {otherwise}. \end{cases}
\end{equation}
Then, due to the active-IRS amplification power constraint, we have \cite{you2021wireless}
\vspace{-0.3cm}
{
\begin{equation}\label{Constraint:power}
\!\!\eta_{s_{\mu(k_{a})}}^{2} (P_{\mathrm{B}}\|  \boldsymbol{\Phi}_{s_{\mu(k_{a})}} \boldsymbol{G}_{0,s_{\mu(k_{a})}} \boldsymbol{w}_{\rm B}\|^{2}
\!+\sigma_{\mathrm{F}}^{2}\sum_{i=1}^{k_{a}-1}\|\boldsymbol{\Phi}_{s_{\mu(k_{a})}}\prod_{\ell=i}^{k_{a}-1}\widetilde{\boldsymbol{G}}_{s_{\mu(\ell)},s_{\mu(\ell+1)}}\|^2+\sigma_{\mathrm{F}}^{2}\|\boldsymbol{\Phi}_{s_{\mu(k_{a})}} \mathbf{I}_{N}\|^{2})\leq P_{s_{\mu(k_{a})}},
\end{equation}
where $P_{s_{\mu(k_{a})}}$ denotes the maximum amplification power of the selected active IRS $s_{\mu(k_{a})}$, and  $\widetilde{\boldsymbol{G}}_{s_{\mu(\ell)},s_{\mu(\ell+1)}}\triangleq\boldsymbol{G}_{s_{\mu(\ell)},s_{\mu(\ell+1)}}\eta_{s_{\mu(\ell)}} \boldsymbol{\Phi}_{s_{\mu(\ell)}}$.}

\noindent Based on the above, the received signal at the user across multiple active and passive IRSs is given by
\vspace{-0.2cm}
\begin{equation}\label{MAMP_signal}
    y_{\rm MAMP}=
     \underbrace{\boldsymbol{g}_{ 0,J+1}^{H}\boldsymbol{w}_{\rm B} x}_{\rm signal} 
    +\underbrace{\sum_{k_{a}=1}^{K_{a}} \boldsymbol{g}^H_{s_{\mu(k_{a})},J+1} \eta_{s_{\mu(k_{a})}}\boldsymbol{\Phi}_{s_{\mu(k_{a})}} \boldsymbol{n}_{s_{\mu(k_{a})}}}_{\rm accumulated \;amplification \; noise} 
    +n_{0},
    \vspace{-6pt}
\end{equation}
where
\vspace{-0.5cm}
\begin{equation}\label{after_active}
\boldsymbol{g}_{s_{\mu(k_{a})},J+1}^{H}\triangleq\begin{cases} \boldsymbol{g}^{H}_{s_{\mu(K_{a})}, J+1} & \text { if } k_{a}=K_{a}\\ 
\boldsymbol{g}^{H}_{s_{\mu(K_{a})}, J+1}\eta_{s_{\mu(K_{a})}}\boldsymbol{\Phi}_{s_{\mu(K_{a})}}   \\\times\left(\prod_{\ell = k_{a}+1}^{K_{a}-1}\boldsymbol{G}_{s_{\mu(\ell)}, s_{\mu(\ell+1)}}\eta_{s_{\mu(\ell)}} \boldsymbol{\Phi}_{s_{\mu(\ell)}}\right) \boldsymbol{G}_{s_{\mu(k_{a})}, s_{\mu(k_{a}+1)}}
& \text { otherwise}. \end{cases}
\end{equation}
denotes the effective cascaded channel from the selected active IRS $s_{\mu(k_{a})}$ to the user.
\vspace{-0.4cm}
\begin{example}
\emph{Fig.~\ref{fig:my_example} illustrates an example of the MAMP-IRS routing path consisting of six active/passive IRSs, which is denoted by $\Omega=\{s_{1}, s_{2}, s_{3}, s_{4},  s_{5}, s_{6}\}$. In particular, two of them are active IRSs (red circles) whose routing-indices are $\{\mu(1),\mu(2)\}=\{2, 4\}$, i.e., $s_{\mu(1)}=s_{2} \in \Omega$ and $s_{\mu(2)}=s_4 \in \Omega$. According to \eqref{Eq:BU_Channel}, the BS$\to$user channel is given by
\vspace{-6pt}
\begin{equation}\label{Eq:K_5_A}
    \boldsymbol{g}_{0,J+1}^H=\boldsymbol{g}^{H}_{s_{4}, J+1}\eta_{s_{4}} \boldsymbol{\Phi}_{s_{4}} 
  \boldsymbol{G}_{s_{2}, s_{4}}\eta_{s_{2}} \boldsymbol{\Phi}_{s_{2}} \boldsymbol{G}_{0, s_{2}},
  \vspace{-6pt}
\end{equation}
where we have $\boldsymbol{g}^{H}_{s_{4}, J+1}=\boldsymbol{h}^H_{s_{6}, J+1}\boldsymbol{\Psi}_{s_{6}}\boldsymbol{H}_{s_{5},s_{6}}\boldsymbol{\Psi}_{s_{5}}\boldsymbol{H}_{s_{4},s_{5}}$, $\boldsymbol{G}_{s_{2}, s_{4}}= \boldsymbol{H}_{s_{3},s_{4}}\boldsymbol{\Psi}_{s_{3}}\boldsymbol{H}_{s_{2},s_{3}}$, and $\boldsymbol{G}_{0, s_{2}} = \boldsymbol{H}_{s_{1},s_{2}} \boldsymbol{\Psi}_{s_{1}}\boldsymbol{H}_{0,s_{1}}$.
Then, we can easily obtain the corresponding signal model according to \eqref{MAMP_signal} and the amplification power constraints in~\eqref{Constraint:power}, which are omitted for brevity.} 
\vspace{-10pt}
\end{example}
\begin{figure}[t]
    \centering
    \includegraphics[width=10cm]{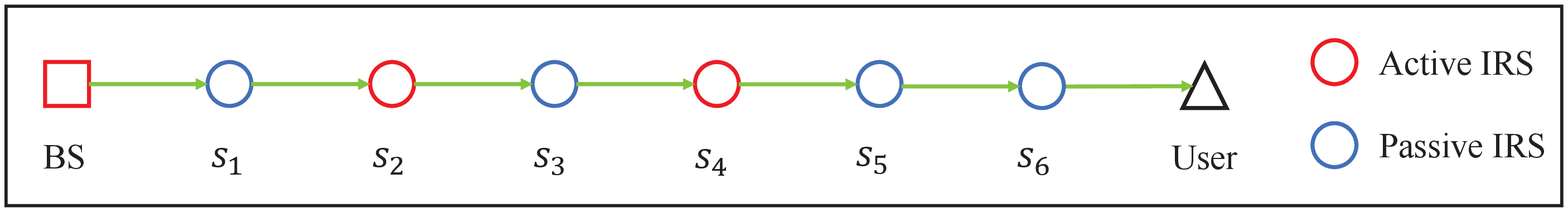}
    \caption{An illustrative example of MAMP-IRS routing path.}
    \label{fig:my_example}
    \vspace{-20pt}
\end{figure}
Based on the signal model in \eqref{MAMP_signal}, the achievable rate of the user aided by MAMP-IRS is given by
$R_{\mathrm{MAMP}}=\log_{2}\left(1+\gamma_{\mathrm{MAMP}}\right)$, where the receive SNR is
\begin{align}
    \gamma_{\rm MAMP} = \frac{P_{\rm B}|\boldsymbol{g}_{ 0,J+1}^{H}\boldsymbol{w}_{\rm B}|^2}{\sum_{k_{a}=1}^{K_{\mathrm{a}}}\| \boldsymbol{g}^H_{s_{\mu(k_{a})},J+1}\eta_{s_{\mu(k_{a})}}\boldsymbol{\Phi}_{s_{\mu(k_{a})}}\|^2 \sigma_{\mathrm{F}}^{2}+\sigma^2}.
\end{align}

Our goal is to maximize the user's achievable rate in the case of MAMP-IRS beam routing by jointly optimizing the multi-reflection routing path, as well as the beamforming of the BS, active IRSs, and passive IRSs. Mathematically, this optimization problem can be formulated as
\begin{subequations}
\begin{align}
\!\!\!\!\!\max_{\substack{\boldsymbol{w}_{\rm B}, \{\boldsymbol{\Psi}_{s_{k_{p}}}\}, \Omega,\\ \{\boldsymbol{\Phi}_{s_{\mu(k_{a})}}\}, \{\eta_{s_{\mu(k_{a})}}\}}}  ~&\log_{2} \left(1+\frac{P_{\rm B}|\boldsymbol{g}_{ 0,J+1}^{H}\boldsymbol{w}_{\rm B}|^2}{\sum_{k_{a}=1}^{K_{\mathrm{a}}}\| \boldsymbol{g}^H_{s_{\mu(k_{a})},J+1}\eta_{s_{\mu(k_{a})}}\boldsymbol{\Phi}_{s_{\mu(k_{a})}}\|^2 \sigma_{\mathrm{F}}^{2}+\sigma^2}\right)
\nn\\
\text{s.t.}~~~~~~~
& \eqref{Eq:Feasible1}-\eqref{Eq:Feasible3},\eqref{Eqs:inter-PIRS}-\eqref{after_active},\nn\\
({\bf P2}):~~~~& K_{a}\ge 1, \label{active_num}\\ &|[\boldsymbol{\Phi}_{s_{\mu(k_{a})}}]_{n,n}|=1,~~\forall n\in\mathcal{N},~s_{\mu(k_{a})} \in \Omega, \label{MAMP_active}
\\
& |[\boldsymbol{\Psi}_{s_{k_{p}}}]_{m,m}|=1,~~\forall m \in\mathcal{M},~s_{k_{p}}\in  \Omega, \label{MAMP_passive}
\\[-1cm]\nn
\end{align}
\end{subequations}
where the constraints~\eqref{Eq:Feasible1}--\eqref{Eq:Feasible3} ensure that the beam routing path is feasible;  \eqref{Constraint:power} enforces the amplification power constraints for the selected active IRSs; $\boldsymbol{g}_{ 0,J+1}^{H}$ and $\boldsymbol{g}^H_{s_{\mu(k_{a})},J+1}$ are given in \eqref{Eq:BU_Channel} and \eqref{after_active}, respectively, which are determined by the routing path $\Omega$ and the joint beamforming of the BS, active IRSs, and passive IRSs; \eqref{active_num} guarantees that there is at least one active IRS in the selected routing path; \eqref{MAMP_active} and \eqref{MAMP_passive} are the unit-modulus constraints imposed by the active and passive IRSs, respectively. Note that problem (P$2$) is a non-convex optimization problem due to the unit-modulus constraints and the intrinsic coupling among $\boldsymbol{w}_{\rm B}$, $ \boldsymbol{\Psi}_{s_{k_{p}}}$, $ \boldsymbol{\Phi}_{s_{\mu(k_{a})}}$, and $\eta_{s_{\mu(k_{a})}}$ in the objective function and path feasibility constraints~\eqref{Eq:Feasible1}--\eqref{Eq:Feasible3}. Moreover, problem (P$2$) is much more challenging than that of the passive-IRS case (i.e., problem (P$1$)), since 1) the selected active IRSs introduce additional amplification power constraints; 2) the active IRSs make the beam routing path design highly coupled; 3) each selected active IRS incurs non-negligible amplification noise 
that is determined by the beam routing path. 

To solve problem (P$2$), we first consider a special SAMP-IRS case in Section~\ref{SAMP} where only one active IRS is involved in the beam routing path and obtain useful insights into the optimal SAMP-IRS beam routing design. Then, we further consider the general case in Section~\ref{MAMP} where multiple active IRSs are involved in the beam routing and propose an efficient algorithm to obtain a near-optimal solution to problem (P$2$).
\vspace{-0.4cm}
\section{Multi-hop beam routing for SAMP-IRS aided wireless communication}\label{SAMP}
\vspace{-0.1cm}

For the SAMP-IRS aided wireless system with one active IRS and ($K-1$) passive IRSs, we denote by $\mu(1)$ the routing-index of the single active IRS, i.e.,  $s_{\mu(1)} \in \mathcal{J}_{a}$. As such, the BS$\to$user equivalent channel given in \eqref{Eq:BU_Channel} can be simplified as
\vspace{-0.5cm}
\begin{align}\label{Eq:SAMPchannel}
\!\!\!&\boldsymbol{g}_{ 0,J+1}^{H}=\boldsymbol{g}^{H}_{s_{\mu(1)}, J+1}  \eta_{s_{\mu(1)}}\boldsymbol{\Phi}_{s_{\mu(1)}}\boldsymbol{G}_{0, s_{\mu(1)}}\nn\\
\!\!\!\!\!\!&=\boldsymbol{h}_{s_{K}, J+1}^{H}\boldsymbol{\Psi}_{s_{K}}\!\!\prod_{k=\mu(1)+1}^{K-1}\left(\boldsymbol{H}_{s_{k},s_{k+1}}\boldsymbol{\Psi}_{s_{k}}\right)\boldsymbol{H}_{s_{\mu(1)},s_{\mu(1)+1}}\eta_{s_{\mu(1)}}\boldsymbol{\Phi}_{s_{\mu(1)}}\!\!
\prod_{k=1}^{\mu(1)-1}\left(\boldsymbol{H}_{s_{k},s_{k+1}} \boldsymbol{\Psi}_{s_{k}}\right)\boldsymbol{H}_{0,s_{1}}.
\end{align}
Moreover, for the SAMP-IRS case, problem (P$2$) reduces to 
\begin{subequations}
\begin{align}
\max_{\substack{\boldsymbol{w}_{\rm B}, \{\boldsymbol{\Psi}_{s_{k}}\}, \boldsymbol{\Phi}_{s_{\mu(1)}}, \eta_{s_{\mu(1)}}, \Omega} }  ~~~ &\log_{2} \left(1+\frac{ P_{\mathrm B}\left| \boldsymbol{g}_{ 0,J+1}^{H} \boldsymbol{w}_{\rm B} \right|^2 }{\| \boldsymbol{g}_{s_{\mu(1)},J+1}^{H} \eta_{s_{\mu(1)}} \boldsymbol{\Phi}_{s_{\mu(1)}}\|^{2} \sigma_{\mathrm{F}}^{2} + \sigma^{2}}\right)
\nn\\
\text{s.t.}~~~~~~~~~~~~
& \eqref{Eq:Feasible1}-\eqref{Eq:Feasible3},\eqref{Eq:SAMPchannel},\nn\\
({\bf P3}):~~~~~~~~& 
\boldsymbol{G}_{0, s_{\mu(1)}}=\prod_{k=1}^{\mu(1)-1}\left(\boldsymbol{H}_{s_{k},s_{k+1}} \boldsymbol{\Psi}_{s_{k}}\right)\boldsymbol{H}_{0,s_{1}},\label{Eq:SAMPbefore}\\
&\boldsymbol{g}_{s_{\mu(1)},J+1}^{H}
    =\boldsymbol{h}_{s_{K}, J+1}^{H}\boldsymbol{\Psi}_{s_{K}}\prod_{k=\mu(1)+1}^{K-1}\left(\boldsymbol{H}_{s_{k},s_{k+1}}\boldsymbol{\Psi}_{s_{k}}\right)\boldsymbol{H}_{s_{\mu(1)},s_{\mu(1)+1}},\label{Eq:SAMPafter}\\
    &\eta^{2}_{s_{\mu(1)}} (P_{\mathrm{B}}\|\boldsymbol{\Phi}_{s_{\mu(1)}} \boldsymbol{G}_{0, s_{\mu(1)}}  \boldsymbol{w}_{\rm B}\|^{2}+\sigma_{\mathrm{F}}^{2}\|\boldsymbol{\Phi}_{s_{\mu(1)}} \mathbf{I}_{N}\|^{2}) \leq P_{s_{\mu(1)}},\label{Eq:Power}\\
&
|[\boldsymbol{\Phi}_{s_{\mu(1)}}]_{n,n}|=1,~~\forall n\in\mathcal{N}, \label{Constraint:Pas_unit}
\\
& |[\boldsymbol{\Psi}_{s_{k}}]_{m,m}|=1,~~\forall m \in\mathcal{M},~s_{k} \in \Omega \setminus s_{\mu(1)}.\label{Constraint:Act_unit}
\end{align}
\end{subequations}
Nonetheless, it is still challenging to solve problem (P$3$) due to the unit-modulus constraints
and  the intrinsic coupling between the two sub-paths (see, e.g.,  the objective function and constraints in \eqref{Eq:SAMPbefore}--\eqref{Eq:Power}). To address this issue, an efficient method is proposed in the next subsection to obtain the optimal solution to problem (P$3$).
\vspace{-0.6cm}
\subsection{Proposed Solution to Problem (P$3$)}
\vspace{-0.2cm}
To solve problem (P$3$), we first optimize the joint beamforming of the BS and IRSs for any given multi-reflection path, based on which we further optimize the SAMP-IRS routing path.

\subsubsection{Joint Beamforming Optimization Given any SAMP-IRS Routing Path}
First, given any feasible SAMP-IRS routing path with the single active IRS involved (i.e., $\Omega$), problem (P$3$) reduces to the following joint beamforming optimization problem
\begin{subequations}
\begin{align}
\ ({\bf P4}): \max_{\substack{\boldsymbol{w}_{\rm B}, \{\boldsymbol{\Psi}_{s_{k}}\}, \boldsymbol{\Phi}_{s_{\mu(1)}}, \eta_{s_{\mu(1)}}} }  &\log_{2} \left(1+\frac{ P_{\mathrm B}\left| \boldsymbol{g}_{ 0,J+1}^{H} \boldsymbol{w}_{\rm B} \right|^2 }{\| \boldsymbol{g}_{s_{\mu(1)},J+1}^{H} \eta_{s_{\mu(1)}} \boldsymbol{\Phi}_{s_{\mu(1)}}\|^{2} \sigma_{\mathrm{F}}^{2} + \sigma^{2}}\right)
\nn\\
\text{s.t.}~~~~~~~~
& \eqref{Eq:SAMPchannel},\eqref{Eq:SAMPbefore}-\eqref{Constraint:Act_unit}.\nn\\[-1.2cm]
\nn
\end{align}
\end{subequations}
The optimal solution to problem (P$4$) is given as follows, which is proved in \cite{zhangSAMP}.
\vspace{-0.4cm}
\begin{lemma}\label{Lem:GivenPath}
\emph{Given any feasible SAMP-IRS routing path $\Omega$, the optimal solution to problem (P$4$) is given by
\vspace{-0.4cm}
\begin{align}
[\boldsymbol{\Psi}_{s_{k}}^*]_{m,m}&=e^{\jmath (\angle [\underline{\boldsymbol{a}_{\rm t}}(s_{k})]_{m}-\angle[\underline{\boldsymbol{a}_{\rm r}}(s_{k})]_{m})}, ~~~\forall m \in \mathcal{M},~s_{k} \in \Omega \setminus s_{\mu(1)},\label{Eq:PIRS}\\
[\boldsymbol{\Phi}_{s_{\mu(1)}}^*]_{n,n}&= e^{\jmath (\angle [\underline{\boldsymbol{a}_{\rm t}}(s_{\mu(1)})]_{n}-\angle[\underline{\boldsymbol{a}_{\rm r}}(s_{\mu(1)})]_{n})}, ~~~\forall n \in \mathcal{N},\label{Eq:AIRS} \\
\boldsymbol{w}_{\rm B}^*&=\boldsymbol{a}_{\rm t} \left(\vartheta^{\rm t}_{0,s_{1}}, U_0\right)/\|\boldsymbol{a}_{\rm t} \left(\vartheta^{\rm t}_{0,s_{1}}, U_0\right)\|,\label{Eq:w}\\
\eta_{s_{\mu(1)}}^*&=\sqrt{\frac{P_{s_{\mu(1)}}}{P_{\mathrm{B}}\|\boldsymbol{\Phi}^*_{s_{\mu(1)}} \boldsymbol{G}_{0, s_{\mu(1)}}   \boldsymbol{w}^*_{\rm B}\|^{2}+\sigma_{\mathrm{F}}^{2}\|\boldsymbol{\Phi}^*_{s_{\mu(1)}} \mathbf{I}_{N}\|^{2}}}, \label{Eq:AF}
\end{align}}
\end{lemma}
\noindent where $\underline{\boldsymbol{a}_{\rm t}}(s_{k}) \triangleq \boldsymbol{a}_{\rm t} (\vartheta^{\rm t}_{s_{k},s_{k+1}}, \theta^{\rm t}_{s_{k},s_{k+1}}, U_{s_{k}})$ and $\underline{\boldsymbol{a}_{\rm r}}(s_{k}) \triangleq \boldsymbol{a}_{\rm r} (\vartheta^{\rm r}_{s_{k},s_{k+1}}, \theta^{\rm r }_{s_{k},s_{k+1}}, U_{s_{k}})$, $s_{k}\in \Omega$ with $\underline{\boldsymbol{a}_{\rm t}}(s_{K}) \triangleq \boldsymbol{a}_{\rm t} (\vartheta^{\rm t}_{{s_{K}},{J+1}}, \theta^{\rm t}_{{s_{K}},{J+1}}, U_{s_{K}})$ and $\underline{\boldsymbol{a}_{\rm r}}(s_{1}) \triangleq \boldsymbol{a}_{\rm r}(\vartheta^{\rm r}_{0,s_{1}}, \theta^{\rm r}_{0,s_{1}}, U_{s_{1}})$.

Lemma~\ref{Lem:GivenPath} shows that the optimal joint beamforming design of the BS and active/passive IRSs are closely determined by the SAMP-IRS routing path (see, \eqref{Eq:PIRS}--\eqref{Eq:AF}). Specifically, the active-IRS amplification factor is determined by the path-loss of the BS$\to$active IRS reflection link. The smaller the path-loss, the smaller the amplification factor. 

\subsubsection{SAMP-IRS Routing Optimization}\label{SAMP-IRS-path}
Given the optimal beamforming design in Lemma \ref{Lem:GivenPath}, we define
\vspace{-0.5cm}
\begin{align}
f_{\rm BA}(\Omega)&\triangleq {\| \boldsymbol{G}_{0, s_{\mu(1)}} \boldsymbol{w}_B^*\|}^2=\frac{M^{2\mu(1)-2} T \beta^{\mu(1)}}{d_{0, s_{1}}^{2} \prod_{k=1}^{\mu(1)-1} d_{s_{k}, s_{k+1}}^{2}}, \label{Eq:BA}\\
f_{\rm AU}(\Omega)&\triangleq {\| \boldsymbol{g}_{s_{\mu(1)},J+1}^{H} \|}^2
=\frac{M^{2K-2\mu(1)} \beta^{K-\mu(1)+1}}{d_{s_{K}, J+1}^{2} d_{s_{\mu(1)}, s_{\mu(1)+1}}^{2} \prod_{k=\mu(1)+1}^{K-1} d_{s_{k}, s_{k+1}}^{2}},\label{Eq:AU}\\
\eta_{s_{\mu(1)}}^2(\Omega)& =\frac{P_{s_{\mu(1)}}}{N (P_{\mathrm{B}} {f}_\mathrm{BA}(\Omega)+\sigma_{\mathrm{F}}^{2})},\label{Eq:eta}
\end{align}
where $f_{\rm BA}(\Omega)$ and $f_{\rm AU}(\Omega)$ represent respectively the end-to-end channel power gain of the BS$\to$active IRS and active IRS$\to$user reflection sub-paths. 
Substituting \eqref{Eq:BA}--\eqref{Eq:eta} into problem (P$3$) yields the following equivalent problem for the SAMP-IRS routing optimization. 
\begin{subequations}
\begin{align}
({\bf P5}):\ \max_{\substack{ \Omega} }  ~~ \frac{P_{\mathrm{B}} N f_{\rm AU}(\Omega)f_{\rm BA}(\Omega)}{f_{\rm AU}(\Omega)  \sigma_{\mathrm{F}}^{2}+\frac{\sigma^{2}(P_{\mathrm{B}} f_{\rm BA}(\Omega)+\sigma_{\mathrm{F}}^{2})}{P_{s_{\mu(1)}}}}
~~~~
\text{s.t.}~~
 \eqref{Eq:Feasible1}-\eqref{Eq:Feasible3}.\nn
\end{align}
\vspace{-16pt}
\end{subequations}

Although much simpler, problem (P$5$) is still challenging to solve due to the coupling between the two sub-paths in the objective function and constraints of \eqref{Eq:Feasible1}--\eqref{Eq:Feasible3}. To deal with this challenge, we first present an important result below.
\vspace{-10pt}
\begin{lemma}\label{Lem:Decouple}
\emph{The solution to problem (P$5$) can be obtained by solving the following two subproblems separately.
\begin{subequations}
\vspace{-0.3cm}
\begin{align}
&({\bf {P6.a}}):\ \ \
\max_{\substack{ \Omega_{\rm BA}} } ~~ \frac{M^{2\mu(1)-2} T \beta^{\mu(1)}}{d_{0, s_{1}}^{2} \prod_{k=1}^{\mu(1)-1} d_{s_{k}, s_{k+1}}^{2}}
~~~~
\text{s.t.}~~~~
 \eqref{Eq:Feasible1}-\eqref{Eq:Feasible3},\nn\\
&({\bf {P6.b}}):\ \ \max_{\substack{\Omega_{\rm AU} } } ~~ \frac{M^{2K-2\mu(1)} \beta^{K-\mu(1)+1}}{d_{s_{K}, J+1}^{2} d_{s_{\mu(1)}, s_{\mu(1)+1}}^{2} \prod_{k=\mu(1)+1}^{K-1} d_{s_{k}, s_{k+1}}^{2}}
~~~~
\text{s.t.}~~~~
 \eqref{Eq:Feasible1}-\eqref{Eq:Feasible3}.\nn
\end{align}
\end{subequations}
where $\Omega_{\rm BA}\triangleq\{s_1, s_2, \ldots, s_{\mu(1)-1}\}$ and $\Omega_{\rm AU}\triangleq 
\{s_{\mu(1)+1}, s_{\mu(1)+2}, \ldots, s_{K}\}$.}
\end{lemma}
\vspace{-0.3cm}
\begin{proof}
First, we consider the effect of $f_{\rm AU}$ on the receive SNR given any $f_{\rm BA}$. Since the receive SNR is given by
$\gamma_{\mathrm{SAMP}}^{(1)}= \frac{P_{\mathrm{B}} N f_{\rm BA}}{  \sigma_{\mathrm{F}}^{2}+\frac{\sigma^{2}(P_{\mathrm{B}} f_{\rm BA}+\sigma_{\mathrm{F}}^{2})}{P_{s_{\mu(1)}}f_{\rm AU}}}$,
it can be easily shown that $\frac{\partial \gamma_{\mathrm{SAMP}}}{\partial f_{\rm AU}} > 0$, and hence maximizing $\gamma_{\mathrm{SAMP}}^{(1)}$ is equivalent to maximizing $f_{\rm AU}$ and  independent of $f_{\rm BA}$. Second,  we consider the effect of $f_{\rm BA}$ given any  $f_{\rm AU}$. The corresponding receive SNR is given by $\gamma_{\mathrm{SAMP}}^{(2)}= \frac{P_{\mathrm{B}}N f_{\rm AU}}{\frac{f_{\rm AU}  \sigma_{\mathrm{F}}^{2}}{f_{\rm BA}}+\frac{\sigma^{2}(P_{\mathrm{B}} +\sigma_{\mathrm{F}}^{2}/f_{\rm BA})}{P_{s_{\mu(1)}}}}$,
for which we can easily show that  $\frac{\partial \gamma_{\mathrm{SAMP}}}{\partial f_{\rm BA}}>0$. This indicates that regardless of the routing path from the active IRS to user, it is always optimal to maximize $f_{\rm BA}$ for maximizing $\gamma_{\mathrm{SAMP}}^{(2)}$.
These lead to the conclusion that maximizing the receive SNR over the multi-reflection path can be equivalently transformed to the channel power gain maximization for the two sub-paths. Combining the above results with \eqref{Eq:BA} and \eqref{Eq:AU} leads to the desired result, thus completing the proof.
\vspace{-8pt}
\end{proof}

\begin{remark}[Decomposable routing path]\label{Rem:RP}
\emph{Lemma~\ref{Lem:Decouple} shows an interesting result that the SAMP-IRS routing design can be equivalently decomposed into two separate sub-path routing designs, which can be intuitively explained as follow. First, the BS$\to$active IRS reflection path determines the incident signal power at the active IRS. The larger the incident signal power, the smaller the effective received noise at the user and hence a higher SNR. Second, the active IRS$\to$user reflection path determines the path-loss of the reflected signal. The smaller the path-loss, the larger the received signal power and hence a higher SNR.}
\vspace{-8pt}
\end{remark}

Next, we solve the optimization problems  (P$6$.a) and  (P$6$.b) by using graph theory \cite{mei2020cooperative}.

$\bullet$ \emph{Problem reformulation:}\label{Sec:Alg}
First, it can be shown that maximizing $f_{\rm BA}(\Omega)$ in problem (P$6$.a) and $f_{\rm AU}(\Omega)$ in problem (P$6$.b) are equivalent to minimizing 
\vspace{-0.3cm}
\begin{align}\label{Eq:inv_G_BA}
\frac{1}{f_{\rm BA}(\Omega)}&=\frac{M^{2}}{T} \cdot \frac{d_{0, s_{1}}^{2}}{M^{2} \beta}   \prod_{k=1}^{\mu(1)-1} \frac{d_{s_{k}, s_{k+1}}^{2}}{M^{2} \beta},\\
\label{Eq:inv_g_AU}
\!\!\!\frac{1}{f_{\rm AU}(\Omega)}\!&=\!M^{2}\cdot\frac{ d_{s_{\mu(1)}, s_{\mu(1)+1}}^{2}}{M^{2} \beta} \!\!\cdot\!\! \frac{d_{s_{K}, J+1}^{2}}{M^{2} \beta} \!\!\!\prod_{k=\mu(1)+1}^{K-1}\!\! \frac{d_{s_{k}, s_{k+1}}^{2}}{M^{2} \beta}.\!\!
\end{align}
Then, by taking the logarithm of \eqref{Eq:inv_G_BA} and \eqref{Eq:inv_g_AU} and ignoring irrelevant constant terms, problems  (P$6$.a) and (P$6$.b) can be reformulated as 
\vspace{-0.3cm}
\begin{align}
\label{P5:shortest}
({\bf P7.a}):~~
~\min _{\Omega_{\rm BA}} ~~&\ln \frac{d_{0, s_{1}}}{M \sqrt{\beta}}+\sum_{k=1}^{\mu(1)-1} \ln \frac{d_{s_{k}, s_{k+1}}}{M \sqrt{\beta}} 
~~~~
\text{s.t.}~~~~
 \eqref{Eq:Feasible1}-\eqref{Eq:Feasible3},\nn\\
({\bf P7.b}):~~
~\min _{\Omega_{\rm AU}} ~~& \ln \frac{d_{s_{\mu(1)}, s_{\mu(1)+1}}}{M \sqrt{\beta}}+\ln \frac{d_{s_{K}, J+1}}{M \sqrt{\beta}}
+\sum_{k=\mu(1)+1}^{K-1} \ln \frac{d_{s_{k}, s_{k+1}}}{M \sqrt{\beta}}
~~~~
\text{s.t.}~~~~
 \eqref{Eq:Feasible1}-\eqref{Eq:Feasible3}.\nn
\end{align}

$\bullet$ \emph{Modified Shortest-path Algorithm:}
As problems (P$7$.a) and (P$7$.b) have similar forms, we only present the algorithm for solving (P$7$.a) in the sequel, while the same method can be used for solving problem (P$7$.b). To be specific, we first recast problem (P$7$.a) into a shortest simple-path problem (SSPP), by constructing a directed weighted graph $\mathcal{G}_{\mathrm{BA}}=(\mathcal{V}_{\mathrm{BA}}, \mathcal{E}_{\mathrm{BA}})$, where the vertex set is defined as $\mathcal{V}_{\mathrm{BA}}\triangleq\{0, s_{\mu(1)}\} \cup \{\mathcal{J}_{\rm BA}\}$ (note that $\mathcal{J}_{\rm BA}$ is a set of nodes located between the BS and the active IRS, which is determined by the network topology). Moreover, there exists an edge between any two vertexes $i,j \in \mathcal{V}_{\mathrm{BA}}$ if  $a_{i,j}=1$ and $d_{j,0} > d_{i,0}$, and hence the edge set is expressed as $\mathcal{E}_{\mathrm{BA}}\triangleq\{(0,j)|a_{0,j}=1, j\in \mathcal{V}_{\mathrm{BA}}\} \cup \{(i,j)|a_{i,j}=1, d_{j,0}>d_{i,0}, i,j\in \mathcal{V}_{\mathrm{BA}}\} \cup \{(j,s_{\mu(1)})|a_{j,s_{\mu(1)}}=1, j\in \mathcal{V}_{\mathrm{BA}}\}$. In addition, the weight between any two nodes $i,j \in \mathcal{V}_{\mathrm{BA}}$ is defined as $W_{i, j}=\ln (d_{i, j}/M \sqrt{\beta}) $. 

Next, problem (P$7$.a) can be solved by using the graph-optimization algorithm. It is worth noting that different from the conventional graph with positive edge weights only, there may exist negative edge weights in the constructed graph $\mathcal{G}_{\mathrm{BA}}$  (i.e., $d_{i, j}<M \sqrt{\beta}$). Thus, we devise an efficient algorithm below to solve problem (P$7$.a) by taking into account the cases with and without negative weights, respectively \cite{mei2020cooperative}. Specifically, if there are no negative weights, the classical Dijkstra algorithm can be directly applied to solve problem (P$7$.a). On the other hand, if there exist negative weights, the classical Dijkstra algorithm cannot guarantee to attain the optimal solution in general as it operates in a greedy manner. Thus, we resort to a recursive algorithm (i.e., dynamic programming) to solve the all hops shortest paths (AHSP) problem~\cite{cheng2004finding}, aiming to find the shortest paths from a given source to any other node in the constructed graph for all possible hop counts. Accordingly, the shortest path of the SSPP can be easily obtained, with the details omitted due to limited space.

\vspace{-0.6cm}
\subsection{Select Active IRS for Beam Routing or Not?}
\vspace{-0.2cm}
In this subsection, we compare the achievable rates of the multi-hop beam routing designs with and without the active IRS involved, based on which we shed key insights into the active-IRS selection for beam routing. 

First, consider the SAMP-IRS beam routing design with the active IRS. Let $f_{\rm BA}^*$ and $f_{\rm AU}^*$  denote the maximum channel power gains of the BS$\to$active IRS and active IRS$\to$user paths, which are determined by the inter-IRS distances and number of passive reflecting elements. Then the corresponding achievable rate is given by 
$R^*_{\mathrm{SAMP}}=\log_{2}(1+\gamma^*_{\mathrm{SAMP}})$, where the maximum SNR is 
\begin{equation}\label{Eq:active}
\gamma^*_{\mathrm{SAMP}}=\frac{P_{\mathrm{B}} N f_{\rm AU}^*  f_{\rm BA}^*}
{ f_{\rm AU}^* \sigma_{\mathrm{F}}^{2} + \frac{\sigma^{2}(P_\mathrm{B} f_{\rm BA}^*+\sigma_{\mathrm{F}}^{2})}{P_{s_{\mu(1)}}}}.
\vspace{-6pt}
\end{equation}
Next, for the passive-IRS beam routing design, we denote by $\tilde{f}_{\rm BU}^*=|(\tilde{\boldsymbol{g}}_{\mathrm{BU}}^{H})^*  \boldsymbol{w}^*_{\rm B}|^2$ the maximum achievable channel power gain \cite{mei2020cooperative}. Then the achievable rate is given by 
$\tilde{R}^*_{\mathrm{pas}}=\log_{2}(1+\tilde{\gamma}^*_{\mathrm{pas}})$, where its maximum SNR is $\tilde{\gamma}^*_{\mathrm{pas}}={P_\mathrm{B} \tilde{f}_{\rm BU}^*}/{\sigma^{2}}$.

Comparing the above leads to the following key result.
\vspace{-6pt}
\begin{theorem}\label{The:Superior}
\emph{For the SAMP-IRS aided wireless communication system, the active IRS should be selected for beam routing (i.e., $R^*_\mathrm{SAMP}\ge \tilde{R}^*_\mathrm{pas}$), if
\begin{align}
\frac{N}{\sigma_{\mathrm{F}}^{2}}  \geq \frac{\tilde{f}_{\rm BU}^*}{f_{\rm BA}^* \sigma^{2}} &+ 
\frac{P_\mathrm{B} \tilde{f}_{\rm BU}^*}{P_{s_{\mu(1)}} f_{\rm AU}^*\sigma_{\mathrm{F}}^{2}}  +\frac{\tilde{f}_{\rm BU}^*}{P_{s_{\mu(1)}} f_{\rm BA}^* f_{\rm AU}^*}.
\vspace{-22pt}
\end{align}}
\end{theorem}
Based on Theorem~\ref{The:Superior}, the effects of active-IRS amplification power and number of active reflecting elements on the active-IRS selection for beam routing are characterized as follows.
\vspace{-0.3cm}
\begin{corollary}[Effect of amplification power]\label{Cor:Power}
\emph{Given the number of active reflecting elements $N$, we have $R^*_\mathrm{SAMP}\ge \tilde{R}^*_\mathrm{pas}$, if the amplification power $P_{s_{\mu(1)}}$ satisfies
\begin{equation}
P_{s_{\mu(1)}} \geq \frac{\tilde{f}_{\rm BU}^* \sigma^{2} (P_\mathrm{B} f_{\rm BA}^* + \sigma_{\mathrm{F}}^{2}) }{f_{\rm AU}^*(N f_{\rm BA}^*\sigma^{2} - \tilde{f}_{\rm BU}^* \sigma_{\mathrm{F}}^{2})}.
\vspace{-6pt}
\end{equation}}
\end{corollary}
Corollary~\ref{Cor:Power} shows that it is preferred to select the active IRS for cooperatively constructing a multi-reflection path if the active-IRS amplification power is sufficiently large. This is intuitively expected since a higher amplification power will lead to a higher amplification factor.
\vspace{-0.3cm}
\begin{corollary}[Effect of number of active reflecting elements]\label{Cor:AE}
\emph{Given the active amplification power $P_{s_{\mu(1)}}$, we have $R^*_\mathrm{SAMP}\ge \tilde{R}^*_\mathrm{pas}$, if the number of active reflecting elements $N$ satisfies
\begin{equation}
 N \geq \frac{\tilde{f}_{\rm BU}^* \sigma_{\mathrm{F}}^{2}}{f_{\rm BA}^*\sigma^{2}} + 
\frac{P_\mathrm{B} \tilde{f}_{\rm BU}^*  }{P_{s_{\mu(1)}} f_{\rm AU}^*} 
+\frac{\tilde{f}_{\rm BU}^* \sigma_{\mathrm{F}}^{2}}{P_{s_{\mu(1)}} f_{\rm BA}^* f_{\rm AU}^*}.
\vspace{-6pt}
\end{equation}}
\end{corollary}

Corollary~\ref{Cor:AE} indicates that it is desirable to select the active IRS for beam routing if the number of active reflecting elements is sufficiently large. This is expected since the achievable rate of the SAMP-IRS beam routing has a power scaling order of $\mathcal{O}(N)$.
\vspace{-6pt}
\begin{remark}[What determines the beam routing path?]\label{Re:Affects}
\emph{
Lemma 2 shows that if the active IRS is involved in the beam routing, the best routing path is determined by the passive IRSs only via e.g., the number of passive reflecting elements $M$ and inter-IRS distances; while it is independent of the configurations (e.g., size and amplification factor) of active IRS. However, the optimal beam routing design in the proposed SAMP-IRS aided wireless system is jointly determined by the active and passive IRSs. Specifically, if the active-IRS amplification power $P_{{s_{\mu(1)}}}$ and/or number of active elements $N$ 
are sufficiently large, the active IRS should be selected for cooperative beam routing and vice versa.}   
\end{remark}
\vspace{-0.7cm}
\section{Multi-hop beam routing for MAMP-IRS aided wireless communication}\label{MAMP}
In this section, we consider the general MAMP-IRS aided wireless communication system where multiple active and passive IRSs cooperatively form a multi-reflection path to assist the downlink communication. For this case, we first characterize the receive SNR for the end-to-end MAMP-IRS beam routing path. Then, we approximate it into a more tractable form and propose an efficient algorithm to obtain a near-optimal solution to problem (P$2$) by using graph theory. The results can be straightforwardly combined with that of the passive-IRS-only case to yield the optimal MAMP-IRS beam routing similarly as in Section~\ref{SAMP-IRS-path}, with the details omitted for brevity.
\vspace{-0.4cm}
\subsection{Proposed Solution to Problem (P$2$)}
\vspace{-0.1cm}
Similar to the case of SAMP-IRS beam routing, to solve problem (P$2$), we first optimize the joint beamforming design for any feasible MAMP-IRS routing path $\Omega$, based on which we further optimize the routing path.
\subsubsection{Joint Beamforming Optimization Given any MAMP-IRS Routing Path}
First, given any feasible MAMP-IRS routing path $\Omega$, problem (P$2$) reduces to the following joint beamforming optimization problem
\begin{subequations}
\begin{align}
({\bf P8}):~~~~
\!\!\!\!\!\max_{\substack{\boldsymbol{w}_{\rm B}, \{\boldsymbol{\Psi}_{s_{k_{p}}}\},\\ \{\eta_{s_{\mu(k_{a})}}\}, \{\boldsymbol{\Phi}_{s_{\mu(k_{a})}}\}} }  &\log_{2} \left(1+\frac{P_{\rm B}|\boldsymbol{g}_{ 0,J+1}^{H}\boldsymbol{w}_{\rm B}|^2}{\sum_{k_{a}=1}^{K_{\mathrm{a}}}\| \boldsymbol{g}^H_{s_{\mu(k_{a})},J+1}\eta_{s_{\mu(k_{a})}}\boldsymbol{\Phi}_{s_{\mu(k_{a})}}\|^2 \sigma_{\mathrm{F}}^{2}+\sigma^2}\right)
\nn\\
\text{s.t.}~~~~~
& \eqref{Eqs:inter-PIRS}-\eqref{after_active},\eqref{active_num}-\eqref{MAMP_passive}.\nn
\\[-1.3cm]\nn
\end{align}
\end{subequations}
The optimal solution to problem~(P8) is obtained by using the similar method as for Lemma~\ref{Lem:GivenPath}.
\vspace{-0.2cm}
\begin{lemma}\label{Lem:MAMP_beamforming}
\emph{Given any feasible MAMP-IRS routing path $\Omega$, the optimal solution to problem (P$8$) is given by
\vspace{-0.5cm}
\begin{align}
&[\boldsymbol{\Psi}_{s_{k_{p}}}^*]_{m,m}= e^{\jmath (\angle [\underline{\boldsymbol{a}_{\rm t}}(s_{k_{p}})]_{m}-\angle[\underline{\boldsymbol{a}_{\rm r}}(s_{k_{p}})]_{m})}, ~~~\forall m \in \mathcal{M},~k_{p} \in \mathcal{K} \setminus \mathcal{K}_{a},\label{Eq:Passive_IRS}\\
&[\boldsymbol{\Phi}_{s_{\mu(k_{a})}}^*]_{n,n}=  e^{\jmath (\angle [\underline{\boldsymbol{a}_{\rm t}}(s_{\mu(k_{a})})]_{n}-\angle[\underline{\boldsymbol{a}_{\rm r}}(s_{\mu(k_{a})})]_{n})}, ~~~\forall n \in \mathcal{N}, ~\mu(k_{a}) \in \mathcal{K}_{a}, \label{Eq:Active_IRS}\\
&\boldsymbol{w}_{\rm B}^*=\boldsymbol{a}_{\rm t} \left(\vartheta^{\rm t}_{0,s_{1}}, U_0\right)/\|\boldsymbol{a}_{\rm t} \left(\vartheta^{\rm t}_{0,s_{1}}, U_0\right)\|,\label{Eq:BSActive}\\
&\eta_{s_{\mu(k_{a})}}^*\!\!=\!\!\!\sqrt{\!\frac{P_{s_{\mu(k_{a})}}}{P_{\mathrm{B}}\|\boldsymbol{\Phi}^*_{s_{\mu(k_{a})}} \boldsymbol{G}_{0,s_{\mu(k_{a})}}  \boldsymbol{w}^*_{\rm B}\|^{2}\!+\!\sum_{i=1}^{k_{a}-1}\|\boldsymbol{\Phi}_{s_{\mu(k_{a})}}\!\prod_{\ell=i}^{k_{a}-1}\widetilde{\boldsymbol{G}}_{s_{\mu(\ell)},s_{\mu(\ell+1)}}\boldsymbol{n}_{s_{\mu(\ell)}}\|^2
		\!+\!\sigma_{\mathrm{F}}^{2}\|\boldsymbol{\Phi}^*_{s_{\mu(k_{a})}} \mathbf{I}_{N}\|^{2}}}, \nn\\ 
&~~~~~~~~~~~~~~~~~~~~~~~~~~~~~~~~~~~~~~~~~~~~~~~~~~~~~~~~~~~~~~~~~~~~~~~~~~~~~~~~~~~~~~\forall \mu(k_{a}) \in \mathcal{K}_{a}.\label{Eq:AmF}
\end{align}}
\end{lemma}

Lemma~\ref{Lem:MAMP_beamforming} indicates that the optimal joint beamforming design of the BS, active IRSs, and passive IRSs are intrinsically determined by the MAMP-IRS routing path (see,  \eqref{Eq:Passive_IRS}--\eqref{Eq:AmF}). 
{
\begin{lemma}\label{Lem:Amplification_approx}
	\emph{For any selected active IRS $s_{\mu(k_{a})}$, its optimal amplification factor given in \eqref{Eq:AmF} can be  approximated as 
	\begin{equation}\label{Eq:Amplification_approx}
	\eta_{s_{\mu(k_{a})}}^*\approx\tilde{\eta}_{s_{\mu(k_{a})}}^*\triangleq\sqrt{\frac{P_{s_{\mu(k_{a})}}}{P_{\mathrm{B}}\|\boldsymbol{\Phi}^*_{s_{\mu(k_{a})}} \boldsymbol{G}_{0,s_{\mu(k_{a})}}  \boldsymbol{w}^*_{\rm B}\|^{2}
				+\sigma_{\mathrm{F}}^{2}\|\boldsymbol{\Phi}^*_{s_{\mu(k_{a})}} \mathbf{I}_{N}\|^{2}}}.
\end{equation}}
\end{lemma}
\begin{proof}
First, the received signal and
noise at the selected active IRS $s_{\mu(k_{a})}$ are defined as $y^{\rm (signal)}_{s_{\mu(k_{a})}}\triangleq P_{\mathrm{B}}\|\boldsymbol{\Phi}^*_{s_{\mu(k_{a})}} \boldsymbol{G}_{0,s_{\mu(k_{a})}}  \boldsymbol{w}^*_{\rm B}\|^{2}$ and $y^{\rm (noise)}_{s_{\mu(k_{a})}}\triangleq \sigma_{\mathrm{F}}^{2}\sum_{i=1}^{k_{a}-1}\|\boldsymbol{\Phi}_{s_{\mu(k_{a})}}\prod_{\ell=i}^{k_{a}-1}\widetilde{\boldsymbol{G}}_{s_{\mu(\ell)},s_{\mu(\ell+1)}}\|^2$, respectively. To prove Lemma~\ref{Lem:Amplification_approx}, we only need to verify that $	\tilde{\eta}_{s_{\mu(k_{a})}}^*/	\eta_{s_{\mu(k_{a})}}^*\approx 1$. To this end, we have 
\begin{align}
\frac{\tilde{\eta}_{s_{\mu(k_{a})}}^*}{\eta_{s_{\mu(k_{a})}}^*}&=\sqrt{\frac{y^{\rm (signal)}_{s_{\mu(k_{a})}}+y^{\rm (noise)}_{s_{\mu(k_{a})}}+\sigma_{\mathrm{F}}^{2}\|\boldsymbol{\Phi}^*_{s_{\mu(k_{a})}} \mathbf{I}_{N}\|^{2}}{y^{\rm (signal)}_{s_{\mu(k_{a})}}+\sigma_{\mathrm{F}}^{2}\|\boldsymbol{\Phi}^*_{s_{\mu(k_{a})}} \mathbf{I}_{N}\|^{2}}}=\sqrt{1+\frac{y^{\rm (noise)}_{s_{\mu(k_{a})}}}{y^{\rm (signal)}_{s_{\mu(k_{a})}}+\sigma_{\mathrm{F}}^{2}\|\boldsymbol{\Phi}^*_{s_{\mu(k_{a})}} \mathbf{I}_{N}\|^{2}}}\nn\\
&=\sqrt{1+\frac{1}{y^{\rm (signal)}_{s_{\mu(k_{a})}}/y^{\rm (noise)}_{s_{\mu(k_{a})}}+\sigma_{\mathrm{F}}^{2}\|\boldsymbol{\Phi}^*_{s_{\mu(k_{a})}} \mathbf{I}_{N}\|^{2}/y^{\rm (noise)}_{s_{\mu(k_{a})}}}}\stackrel{(a_{1})}{\approx}1,
\end{align}
where $(a_{1})$ holds since $y^{\rm (signal)}_{s_{\mu(k_{a})}}/y^{\rm (noise)}_{s_{\mu(k_{a})}}=\frac{P_{\rm B}}{\sigma_{\mathrm{F}}^{2}}\frac{\|\boldsymbol{\Phi}^*_{s_{\mu(k_{a})}} \boldsymbol{G}_{0,s_{\mu(k_{a})}}  \boldsymbol{w}^*_{\rm B}\|^{2}}{\sum_{i=1}^{k_{a}-1}\|\boldsymbol{\Phi}_{s_{\mu(k_{a})}}\prod_{\ell=i}^{k_{a}-1}\widetilde{\boldsymbol{G}}_{s_{\mu(\ell)},s_{\mu(\ell+1)}}\|^2}\gg 1$.

\noindent Combining the above steps leads to the desired result.
\end{proof}}

\subsubsection{MAMP-IRS Routing Optimization}
With the optimal joint beamforming design in Lemma~\ref{Lem:MAMP_beamforming}, we then define
\vspace{-0.3cm}
\begin{align}
f_{0, s_{\mu(1)}}(\Omega)&\triangleq \|\boldsymbol{G}_{0, s_{\mu(1)}}\boldsymbol{w}^*_{\rm B}\|^2 =\left\|\prod_{k=1}^{\mu(1)-1}\left(\boldsymbol{H}_{s_{k},s_{k+1}} \boldsymbol{\Psi}^*_{s_{k}}\right) \boldsymbol{H}_{0,s_{1}}\boldsymbol{w}^*_{\rm B}\right\|^2, \label{gain_01}\\ 
f_{s_{\mu(k_{a})}, s_{\mu(k_{a}+1)}}(\Omega)&\triangleq\|\boldsymbol{G}_{s_{\mu(k_{a})}, s_{\mu(k_{a}+1)}}\|^2=\left\|\prod_{k=\mu(k_{a})+1}^{\mu(k_{a}+1)-1}\left(\boldsymbol{H}_{s_{k},s_{k+1}}\boldsymbol{\Psi}^*_{s_{k}}\right)\boldsymbol{H}_{s_{\mu(k_{a})},s_{\mu(k_{a})+1}}\right\|^2, \label{gain_kk1} \\ 
f_{s_{\mu(K_{a})},J+1}(\Omega)&\triangleq\|\boldsymbol{g}^{H}_{s_{\mu(K_{a})}, J+1}\|^2\!\!=\!\!\left\|\boldsymbol{h}_{s_{K}, J+1}^{H}\boldsymbol{\Psi}^*_{s_{K}}
\!\!\!\prod_{k=\mu(K_{a})+1}^{K-1}\!\!\!\left(\boldsymbol{H}_{s_{k},s_{k+1}}\!\!\boldsymbol{\Psi}^*_{s_{k}}\right)\boldsymbol{H}_{s_{\mu(K_{a})},s_{\mu(K_{a})+1}}\right\|^2, \!\!\!\!\!\!\label{gain_kJ}\\ 
\eta_{s_{\mu(k_{a})}}^{2}(\Omega)&\triangleq \frac{P_{s_{\mu(k_{a})}}}{N(P_{\mathrm{B}}f_{0, s_{\mu(1)}}(\Omega)\prod_{t=1}^{k_{a}-1}(f_{s_{\mu(k_{a})}, s_{\mu(k_{a}+1)}}(\Omega)\eta^{2}_{s_{\mu(t)}}N^2)
+\sigma_{\mathrm{F}}^{2})}, \label{amplify}
\end{align}
where $f_{0, s_{\mu(1)}}(\Omega)$, $f_{s_{\mu(k_{a})}, s_{\mu(k_{a}+1)}}(\Omega)$, and $f_{s_{\mu(K_{a})},J+1}(\Omega)$ denote respectively the end-to-end channel power gains of the BS$\to$the first selected active IRS, $k_{a}$-th active IRS$\to$($k_{a}$+1)-th active IRS, and $K_{a}$-th active
IRS$\to$user reflection sub-paths. Substituting \eqref{gain_01}--\eqref{gain_kJ} into problem (P$2$) yields the following equivalent form for the MAMP-IRS routing optimization.
\vspace{-0.3cm}
\begin{subequations}
\begin{align}
({\bf P9}):~~~~
\max_{\substack{\Omega}}  ~~& \frac{P_{\rm B}N^{2K_{a}}f_{s_{\mu(K_{a})},J+1}(\Omega)\eta^{2}_{s_{\mu(K_{a})}}\left[\prod_{k_{a}=1}^{K_{a}-1}f_{s_{\mu(k_{a})}, s_{\mu(k_{a}+1)}}(\Omega)\eta^{2}_{s_{\mu(k_{a})}}\right]f_{0, s_{\mu(1)}}(\Omega)}
{\sum_{k_{a}=1}^{K_{\mathrm{a}}}\left[f_{s_{\mu(K_{a})},J+1}(\Omega)N\eta^{2}_{s_{\mu(K_{a})}}\prod_{\ell=k_{a}}^{K_{a}-1}f_{s_{\mu(\ell)}, s_{\mu(\ell+1)}}(\Omega)\eta^{2}_{s_{\mu(\ell)}}N \right]\sigma_{\mathrm{F}}^{2}+\sigma^2}
\nn\\
\text{s.t.}~~
& \eqref{Eq:Feasible1}-\eqref{Eq:Feasible3},\eqref{active_num},\eqref{gain_01}-\eqref{amplify}.\nn
\\[-1.3cm]\nn
\end{align}
\end{subequations}
Note that problem (P$9$) is a combinatorial optimization problem due to the discrete optimization variables (i.e., $\Omega$), which is difficult to solve optimally in general. To address this issue, we first present an important lemma as follows.
\vspace{-12pt}
\begin{lemma}\label{Le:General}
\emph{Given any active-IRS beam routing path that specifies the indices of active IRSs in $\Omega$, i.e., $\Omega_{\rm act}=\{s_{\mu(1)},s_{\mu(2)},\cdots, s_{\mu(K_{a})}\}$, problem (P$9$) reduces to $K_{a}+1$ subproblems, corresponding to the passive-IRS beam routing optimization for maximizing the effective channel power gain of the sub-paths from the BS to the first active IRS $s_{\mu(1)}$ (given in \eqref{gain_01}), from the last active IRS $s_{\mu(K_{a})}$ to the user (given in \eqref{gain_kJ}), as well as that between each two adjacent (selected)
active IRSs along $\Omega_{\rm act}$ (given in \eqref{gain_kk1}).}
\end{lemma}
\begin{proof}
 \upshape See Appendix \ref{App_2}.
\end{proof}

Lemma~\ref{Le:General} can be regarded as a general version of Lemma~\ref{Lem:Decouple}, which shows that given any active-IRS routing path $ \Omega_{\rm act}$, the complicated MAMP-IRS routing optimization problem (i.e., problem (P$9$)) can be equivalently decomposed into a set of independent subproblems for the 
passive-IRS beam routing. More importantly, it also indicates that regardless of the active-IRS routing path, the optimal passive-IRS routing path between each two active IRSs is always the one that maximizes the effective cascaded channel power gain. These facts thus motivate us to develop a new hierarchical MAMP-IRS beam routing design for solving problem (P$9$), with two key procedures elaborated as follows.
\begin{itemize}
\item \textbf{\underline{Phase 1 (Inner passive-IRS beam routing)}}:
The first phase aims to determine the optimal passive-IRS routing paths between each two active IRSs and that between the BS/user and each active IRS. Take the passive-IRS routing between two active IRSs as an example. {\color{black} If there exists a feasible path (consisting of passive IRSs only) between active IRSs $i$ and $j$ with $i, j\in \mathcal{J}_{a}$, its inner passive-IRS routing optimization problem can be formulated as 
\begin{subequations}
\begin{align}
({\bf P10}):\ \max_{\substack{ \tilde{\Omega}_{i,j}} }  ~~ f_{i, j}(\tilde{\Omega}_{i,j})~~~~~~
\text{\rm s.t.}~~
 \eqref{Eq:Feasible1}-\eqref{Eq:Feasible3},\nn
\vspace{-3pt}
\end{align}
\end{subequations}
where $\tilde{\Omega}_{i,j}\triangleq\{s_{1},s_{2},\cdots,s_{\tilde{K}_{i,j}}\}$ denotes the passive-IRS routing path between the two active IRSs with $\tilde{K}_{i,j}$ denoting the number of passive IRSs in the path. Moreover,  $f_{i,j}(\tilde{\Omega}_{i,j})$ is the effective channel power gain between active IRSs $i$ and $j$, which is given by
\begin{equation}
   f_{i,j}(\tilde{\Omega}_{i,j}) \triangleq \frac{M^{2\tilde{K}_{i,j}} \beta^{\tilde{K}_{i,j}+1}}{ d_{i, s_{1}}^{2}d_{s_{\tilde{K}_{i,j}}, j}^{2}\prod_{\tilde{k}=1}^{\tilde{K}_{i,j}-1} d_{s_{\tilde{k}}, s_{\tilde{k}+1}}^{2}}.\label{Eq:f_ij}
\end{equation}
The optimal solution to problem (P$10$) can be obtained by following the similar procedures in Section~\ref{SAMP-IRS-path}, which is denoted by $f^*_{i,j}$.}
\item \textbf{\underline{Phase 2 (Outer active-IRS beam routing)}}: Given the optimal inner passive-IRS beam routing designed in Phase 1 as well as the corresponding effective channel power gains of the paths among active IRSs (i.e., $f^*_{i,j}$), the second phase aims to optimize the outer active-IRS routing path for maximizing the achievable rate. This problem is  formulated as
\end{itemize}
\vspace{-0.2cm}
\begin{subequations}
\begin{align}
({\bf P11}):~~~~
\max_{\substack{\Omega_{\rm act}}}  ~~& \frac{P_{\rm B}N^{2K_{a}}f^*_{s_{\mu(K_{a})},J+1}\eta^{2}_{s_{\mu(K_{a})}}\left(\prod_{k_{a}=1}^{K_{a}-1}f^*_{s_{\mu(k_{a})}, s_{\mu(k_{a}+1)}}\eta^{2}_{s_{\mu(k_{a})}}\right)f^*_{0, s_{\mu(1)}}}
{\sum_{k_{a}=1}^{K_{\mathrm{a}}}\left(f^*_{s_{\mu(K_{a})},J+1}N\eta^{2}_{s_{\mu(K_{a})}}\prod_{\ell=k_{a}}^{K_{a}-1}f^*_{s_{\mu(\ell)}, s_{\mu(\ell+1)}}\eta^{2}_{s_{\mu(\ell)}}N \right)\sigma_{\mathrm{F}}^{2}+\sigma^2}
\label{SNR_MAMP}\\
\text{s.t.}~~~
& \eqref{Eq:Feasible1}-\eqref{Eq:Feasible3},\eqref{active_num},\eqref{amplify}.\nn
\\[-1.3cm]\nn
\end{align}
\end{subequations}
Problem (P$11$) is still challenging to solve since the objective function is highly complicated due to the intrinsic coupling of $\Omega_{\rm act}$ in the objective function and constraint \eqref{amplify}. Hence, to make the problem more tractable, an explicit expression for the receive SNR is presented below.
\vspace{-0.1cm}
{
\begin{lemma}\label{Lem:SNR_Exp}
\emph{
Given the optimal joint beamforming in Lemma~\ref{Lem:MAMP_beamforming} and the maximum effective channel power gains between each two active IRSs in \eqref{Eq:f_ij}, the receive SNR at the user over the MAMP-IRS routing path given in \eqref{SNR_MAMP} can be equivalently expressed as 
\vspace{-0.3cm}
\begin{equation}\label{SNR_out}
    \gamma_{\rm MAMP}=\frac{N^{K_{a}}}{\sum_{k_{a}=1}^{K_{a}+1} \sum_{t=1}^{k_{a}} \prod_{i=t}^{k_{a}} \frac{1}{ \gamma_{i}}\prod_{j=1}^{t-2}N},
    \vspace{-6pt}
\end{equation}
where $\gamma_{i}$ represents the per-hop SNR between nodes $(i-1)$ and $ i \in \bar{\Omega}_{\rm act}\triangleq\{\Omega_{\rm act},J+1\}$, defined as
\begin{equation}\label{SNR_DEFINITION}
\gamma_{i}= \begin{cases} P_{\rm B}f^*_{0, s_{\mu(1)}} / \sigma^2_{\mathrm{F}} & \text { if } i=1 \\   P_{s_{\mu(i-1)}} f^*_{s_{\mu(i-1)}, s_{\mu(i)}} / \sigma^2_{\mathrm{F}} & \text { if } 2 \leq i \leq K_{a} \\  P_{s_{\mu(K_{a})}} f^*_{s_{\mu(K_{a})},J+1} / \sigma^2 &  \text { if } i = J+1.\end{cases}
\end{equation}}
\end{lemma}}
\begin{proof}
 \upshape See Appendix \ref{App_1}.
\end{proof}
Lemma~\ref{Lem:SNR_Exp} demonstrates a fundamental tradeoff between maximizing the prominent multiplicative
active-IRS beamforming gain (see the numerator in \eqref{SNR_out}) versus minimizing the accumulated amplification noise power (see the denominator in \eqref{SNR_out}). 

By using Lemmas~\ref{Lem:MAMP_beamforming}--\ref{Lem:SNR_Exp}, problem (P$11$) can be equivalently transformed into the following problem
\vspace{-0.4cm}
{
\begin{equation}
({\bf P12}):~~\max_{\substack{ \Omega_{\rm act}}}  ~~ \frac{N^{K_{a}}}{\sum_{k_{a}=1}^{K_{a}+1} \sum_{t=1}^{k_{a}} \prod_{i=t}^{k_{a}} \frac{1}{ \gamma_{i}}\prod_{j=1}^{t-2}N}~~~~
\text{s.t.}~~~
 \eqref{Eq:Feasible1}-\eqref{Eq:Feasible3}, \eqref{active_num}.\nn
 \vspace{-5pt}
\end{equation}
}
\begin{figure}[t]
	\centering
	\includegraphics[height=6cm]{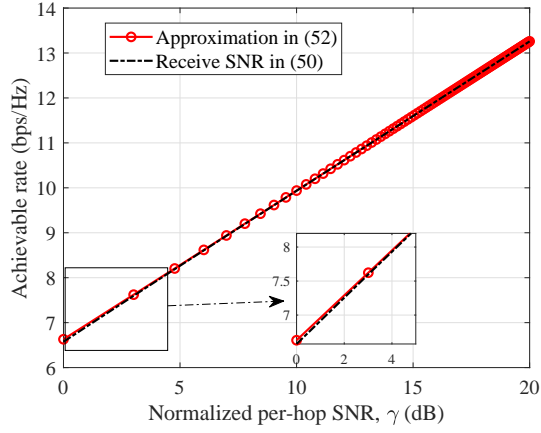}
	\caption{Accuracy of the receive SNR lower bound given in \eqref{Eq:lowerb}.}\label{fig:verify_appro}
	\vspace{-18pt}
\end{figure}

Note that the denominator of the objective function is a complicated function of the per-hop SNR of all selected active IRSs, which is determined by the active-IRS routing path $\Omega_{\rm act}$. To tackle this difficulty, we first obtain a more tractable lower bound for the receive SNR given in \eqref{SNR_out}, which shall facilitate the design of a near-optimal MAMP-IRS routing path.
\vspace{-5pt}
{
\begin{theorem}\label{theorem:L_Bound} 
\emph{
For the MAMP-IRS routing path $\Omega$, in the high-SNR regime, the receive SNR given in~\eqref{SNR_out} is approximated as
\vspace{-0.4cm}
\begin{equation}\label{Eq:lowerb}
\gamma_{\mathrm{MAMP}} \approx \gamma^{\mathrm{(app)}}_{\mathrm{MAMP}} \triangleq N^{K_{a}}\left[\sum_{k_{a}=1}^{K_{a}+1} \frac{1}{ \gamma_{k_{a}}}\prod_{j=1}^{k_{a}-2}N\right]^{-1}.
\vspace{-0.2cm}
\end{equation}}
\end{theorem}

\begin{proof}
First, we define $\xi_{\mathrm{MAMP}}\triangleq \sum_{k_{a}=1}^{K_{a}+1} \sum_{t=1}^{k_{a}} \prod_{i=t}^{k_{a}} \frac{1}{ \gamma_{i}}\prod_{j=1}^{t-2}N$ and $\xi^{\mathrm{(app)}}_{\mathrm{MAMP}}\triangleq \sum_{k_{a}=1}^{K_{a}+1} \frac{1}{ \gamma_{k_{a}}}\prod_{j=1}^{k_{a}-2}N$. As such, to prove Theorem~\ref{theorem:L_Bound}, we only need to show that $\delta \triangleq \lim_{\gamma\to\infty} \frac{\xi^{\mathrm{(app)}}_{\mathrm{MAMP}}}{\xi_{\mathrm{MAMP}}}\approx 1$, defined as 
\begin{align}
	\delta \triangleq \lim_{\gamma\to\infty} \frac{\xi^{\mathrm{(app)}}_{\mathrm{MAMP}}}{\xi_{\mathrm{MAMP}}}&=\lim_{\gamma\to\infty}\frac{\xi_{\mathrm{MAMP}}-\sum_{k_{a}=1}^{K_{a}} \frac{1}{ \gamma_{k_{a}}\gamma_{k_{a}+1}}\prod_{j=1}^{k_{a}-2}N}{\xi_{\mathrm{MAMP}}}\nn\\
	&=\lim_{\gamma\to\infty}\left(1-\frac{\sum_{k_{a}=1}^{K_{a}} \frac{1}{ \gamma_{k_{a}}\gamma_{k_{a}+1}}\prod_{j=1}^{k_{a}-2}N}{\xi_{\mathrm{MAMP}}}\right).
\end{align}
Next, in the high-SNR regime, from the definition of per-hop SNR in \eqref{SNR_DEFINITION}, one can easily conclude that: 1) $\frac{1}{\gamma_{i}} \ll 1, \; \forall i \in \bar{\Omega}_{\rm act}$; 2) $\gamma_{i}$ and $\gamma_{j}$ with $ i, j \in \bar{\Omega}_{\rm act}$ and $i\neq j$ are independent of each other; 3) $\frac{1}{\gamma_{i}} \gg \frac{1}{\gamma_{i}\gamma_{j}}, \; \forall i, j \in \bar{\Omega}_{\rm act}, i \neq j$.
As a result, it can be shown that $\delta \approx 1$,  
thus leading to the desired result in \eqref{Eq:lowerb}.
\end{proof}
}

Theorem~\ref{theorem:L_Bound} is intuitively expected since the receive SNR in \eqref{Eq:lowerb} is an approximation of \eqref{SNR_out} by dropping several non-dominant terms to the denominator, which thus have a negligible effect in the high-SNR regime. To show this, we numerically illustrate the accuracy of the obtained receive-SNR approximation in Fig.~\ref{fig:verify_appro} under the system setup with 4 reflection hops (i.e., $K_{a}=3$) and $N=100$. It is observed that the achievable rate based on the SNR approximation in \eqref{Eq:lowerb} is very close to that of the exact receive SNR, when the per-hop SNR is larger than $5$ dB. Moreover, it is worth noting that the high-SNR condition can be achieved easily in the considered MAMP-IRS aided wireless system. For example, for an indoor system with $P_{\rm B}=20$ dBm, $\beta=-46$ dB, $d=12$ m, $\sigma^2_{\rm F}=-70$ dBm, and $\sigma^2=-80$ dBm, the corresponding per-hop SNR can easily reach $20$ dB. While for the low-SNR regime, i.e., $\gamma<5$ dB, the obtained receive-SNR approximation may lead to the modest overestimation of the actual
achievable rate. 

With Theorem~\ref{theorem:L_Bound}, problem (P$12$) can be approximated as follows by replacing the objective function with its lower bound.
\vspace{-0.4cm}
\begin{subequations}
\begin{align}
({\bf P13}):~~\max_{\substack{ \Omega_{\rm act}}}  ~~ \gamma^{\mathrm{(app)}}_{\mathrm{MAMP}}~~~~
\text{s.t.}~~~
 \eqref{Eq:Feasible1}-\eqref{Eq:Feasible3}, \eqref{active_num}.\nn
\end{align}
\end{subequations}
In the following, we propose an efficient algorithm to solve problem (P$13$) by using graph theory similarly as in Section~\ref{Sec:Alg}.

{First,  maximizing $\gamma^{\mathrm{(app)}}_{\mathrm{MAMP}}$ in problem (P$13$) is equivalent to minimizing
\begin{equation}\label{Eq:or_1/L_SNR}
\frac{1}{\gamma^{\mathrm{(app)}}_{\mathrm{MAMP}}}=\frac{\sum_{k_{a}=1}^{K_{a}+1} \frac{1}{ \gamma_{k_{a}}}\prod_{j=1}^{k_{a}-2}N}{ N^{K_{a}}}. 
\end{equation}
Next, substituting \eqref{SNR_DEFINITION} into \eqref{Eq:or_1/L_SNR} yields
\begin{equation}\label{Eq:1/L_SNR}
\!\!\!\frac{1}{\gamma^{\mathrm{(app)}}_{\mathrm{MAMP}}}=\frac{1}{N^{K_{a}}}\underbrace{\left[\left(\frac{\sigma^2_{\rm F}}{P_{\rm B}f^*_{0, s_{\mu(1)}}}\right)+\left(\frac{\sigma^2\prod_{t=1}^{K_{a}-1}N}{P_{s_{\mu(K_{a})}}f^*_{s_{\mu(K_{a})}, J+1}}\right)+\sum_{k_{a}=2}^{K_{a}}\left(\frac{\sigma^2_{\rm F}\prod_{t=1}^{k_a-2}N}{P_{s_{\mu(k_{a}-1)}}f^*_{s_{\mu(k_{a}-1)}, s_{\mu(k_{a})}}}\right)\right]}_{
\varphi (K_{a},\{s_{\mu(k_{a})}\}^{K_{a}}_{k_{a}=1})}
\end{equation}
Then, problem (P$13$) can be reformulated as 
\begin{subequations}
\begin{align}
({\bf P14}):~~~~\min_{\substack{ \Omega_{\rm act}}}  ~~
\frac{1}{N^{K_{a}}} \times \varphi(K_{a},\{s_{\mu(k_{a})}\}^{K_{a}}_{k_{a}=1})~~~~
\text{s.t.}~~
 \eqref{Eq:Feasible1}-\eqref{Eq:Feasible3}, \eqref{active_num}.\nn
 \\[-1.3cm]\nn
\end{align}
\end{subequations}
\begin{lemma}\label{Problemtrans}
\emph{The solution to problem (P$14$) can be obtained by solving the following problem
\begin{subequations}
\begin{align}
({\bf P15}):~~~~\min_{\substack{ K_{a}}}~~   \left[\frac{1}{N^{K_{a}}}\times[\min_{\substack{\{s_{\mu(k_{a})}\}^{K_{a}}_{k_{a}=1}}}  ~\varphi(K_{a},\{s_{\mu(k_{a})}\}^{K_{a}}_{k_{a}=1})]\right]~~~~
\text{s.t.}~~
 \eqref{Eq:Feasible1}-\eqref{Eq:Feasible3}, \eqref{active_num}.\nn
\end{align}
\end{subequations}}
\end{lemma}
\begin{proof}
Let $\Omega^*_{\rm act}$ denote the optimal active-IRS routing path with $K^*_{a}$ selected active IRSs. Thus, it can be easily shown that minimizing $1/\gamma^{\mathrm{(low)}}_{\mathrm{MAMP}}$ is equivalent to the following two-layer problem: 1) the inner layer is to optimize the active-IRS routing path $\Omega_{\rm act}$ for minimizing $\varphi(K_{a},\{s_{\mu(k_{a})}\}^{K_{a}}_{k_{a}=1})$ given any feasible $K_{a}$; 2) the outer layer is to search the optimal $K_{a}$ for minimizing $1/N^{K_{a}}\times \varphi(K_{a},\{s_{\mu(k_{a})}\}^{K_{a}}_{k_{a}=1})$.
\end{proof}

To solve the inner problem, we transform the inner problem into the following form
\begin{align}
({\bf P16}):\min_{\{s_{\mu(k_{a})}\}^{K_{a}}_{k_{a}=1}} ~&\frac{\sigma^2_{\rm F}}{P_{\rm B}f^*_{0, s_{\mu(1)}}}+\frac{\sigma^2\prod_{t=1}^{K_{a}-1}N}{P_{s_{\mu(K_{a})}}f^*_{s_{\mu(K_{a})}, J+1}}+\sum_{k_{a}=2}^{K_{a}}\frac{\sigma^2_{\rm F}\prod_{t=1}^{k_a-2}N}{P_{s_{\mu(k_{a}-1)}}f^*_{s_{\mu(k_{a}-1)}, s_{\mu(k_{a})}}}
~~~~\text{s.t.}~
\eqref{Eq:Feasible1}-\eqref{Eq:Feasible3}.\nn
\end{align}}
It is observed that problem (P$16$) has a similar form with problems (P$7$.a) and (P$7$.b). Thus, the proposed graph-optimization algorithms in Section~\ref{SAMP-IRS-path} can be readily applied to solve problem (P$16$). Next, the outer problem can be solved by an exhaustive search for the optimal number of active IRSs.
\vspace{-0.4cm}
\begin{remark}[Computation Complexity Analysis]
\emph{{\color{black}
First, for Phase 1, we say there is an inter-active-IRS edge between each two active IRSs (and between the BS/user and active IRS) if there exists a feasible all-passive-IRSs routing path in between. The total number of inter-active-IRS edges is denoted by $E_{\rm act}$. Then, for the passive-IRS routing design between each two active IRSs, its worst-case complexity is $\mathcal{O}(J_{p}|\mathcal{E}|)$ where $\mathcal{E}$ denotes the edge set of the whole graph. Hence, the worst-case complexity of Phase 1 is $\mathcal{O}(E_{\rm act} J_{p}|\mathcal{E}|)$ \cite{mei2020cooperative}.
 Next, the complexity of Phase 2 depends on the inner-layer and outer-layer problems. For the inner-layer active-IRS beam routing optimization, its routing graph can be represented by $\mathcal{G}_{0,J+1}=(\mathcal{V}_{0,J+1},\mathcal{E}_{0,J+1})$, where $\mathcal{V}_{0,J+1}$ denotes the set of the BS, user, and active IRSs; and $\mathcal{E}_{0,J+1}$ denotes the inter-active-IRS edge set. 
 Given the number of selected active IRSs, $K_{a}$, the worst-case complexity for solving the inner problem is $\mathcal{O}(K_{a}|\mathcal{E}_{0,J+1}|)$ \cite{cheng2004finding}. 
 Accounting for the outer-layer search for the optimal $K_{a}$, the worst-case complexity of Phase 2 is $\mathcal{O}(\sum_{K_{a}=1}^{J_{a}} K_{a}|\mathcal{E}_{0,J+1}|)=\mathcal{O}(J_{a}^2|\mathcal{E}_{0,J+1}|)$. To summarize, the overall (worst-case) computation complexity of the proposed MAMP-IRS beam routing algorithm is  $\mathcal{O}(E_{\rm act} J_{p}|\mathcal{E}|+J_{a}^2|\mathcal{E}_{0,J+1}|)$.
 }}
\end{remark}
\vspace{-0.8cm}
\section{Numerical Results}\label{Results}
\vspace{-0.2cm}
In this section, we present numerical results to demonstrate the effectiveness of the proposed MAMP-IRS beam routing scheme, while the SAMP-IRS case is its special case whose results can be obtained by setting $K_{a}=1$ (see the conference version \cite{zhangSAMP} for more details). 
The simulation setup is as follows unless otherwise specified. We set
$P_{\rm B}=20$ dBm, $T=4$, $\lambda=0.06$ m,  $\beta=(\lambda/4\pi)^{2} = -46$ dB,  $\sigma^{2}= -80$ dBm, and $\sigma_{\mathrm{F}}^{2} = -70$ dBm \cite{you2021wireless}. Moreover, we consider the same amplification power for all active IRSs, i.e., $P_{\ell}=P_{\rm F}=10$ dBm, $\forall~\ell\in \mathcal{J}_{a}$. As shown in Fig.~\ref{High_appro_route}, three active IRSs (marked as red circles), i.e., ${J_{a}=3}$, and seven passive IRSs (marked as blue circles) are deployed in  an indoor environment (e.g., smart factory), where the communication links with available LoS paths are represented by edges.
\vspace{-0.5cm}
\subsection{Optimized Beam Routing Path}
{First, we compare in Fig.~\ref{High_appro_route} the optimized MAMP-IRS beam routing paths by the proposed scheme and the optimal one based on the exhaustive search, given different numbers of elements at each passive IRS, $M$, and at each active IRS, $N$. Several important
observations are made as follows. First, it is observed that the proposed MAMP-IRS beam routing scheme yields the same beam routing path with the optimal scheme. Second, as $M$ increases from $M=500$ to $M=1500$, more passive IRSs are involved in the passive-IRS beam routing path. This is because when $M$ is larger, the passive beamforming gain brought by traveling through more passive IRSs is more dominant than the increased product-distance path-loss. Moreover, it is observed from Figs.~\ref{Route1500_100} and \ref{Route1500_1000} that, for a sufficiently large $M$ (i.e., $M=1500$), both beam routing paths go through only two active IRSs (i.e., red circles) regardless of the number of active reflecting elements $N$. This is expected, since when the per-hop SNR is sufficiently large,  selecting an active IRS for beam routing may not tend to provide a higher active beamforming gain and a power amplification gain in the considered MAMP-IRS aided system.}
\begin{figure}[t]
\centering
\subfigure[$N=100, M=500$]{\label{Route500}
	\includegraphics[width=8cm]{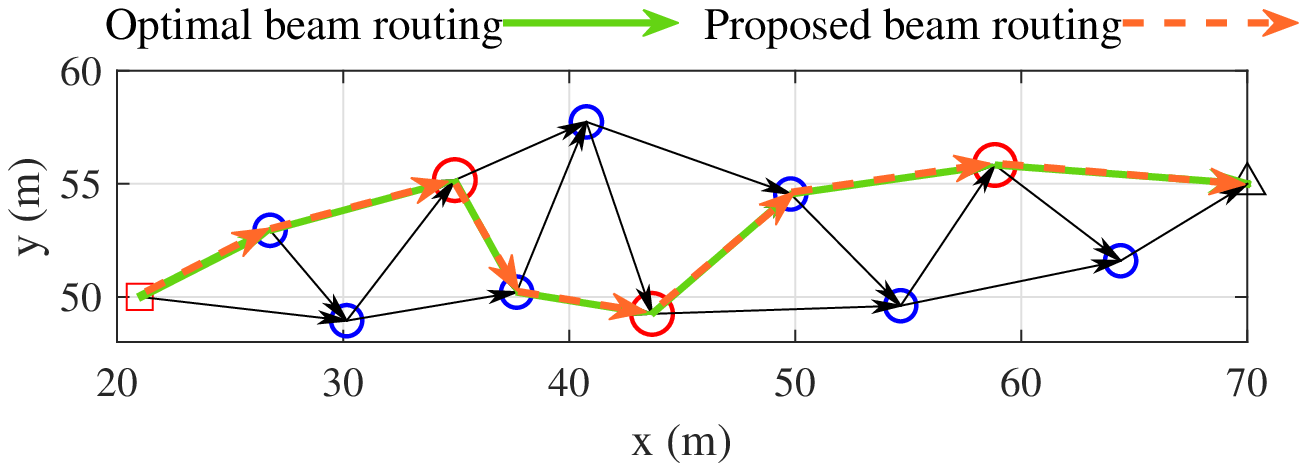}}
\subfigure[$N=100, M=1000$]{\label{Route1000}
	\includegraphics[width=8cm]{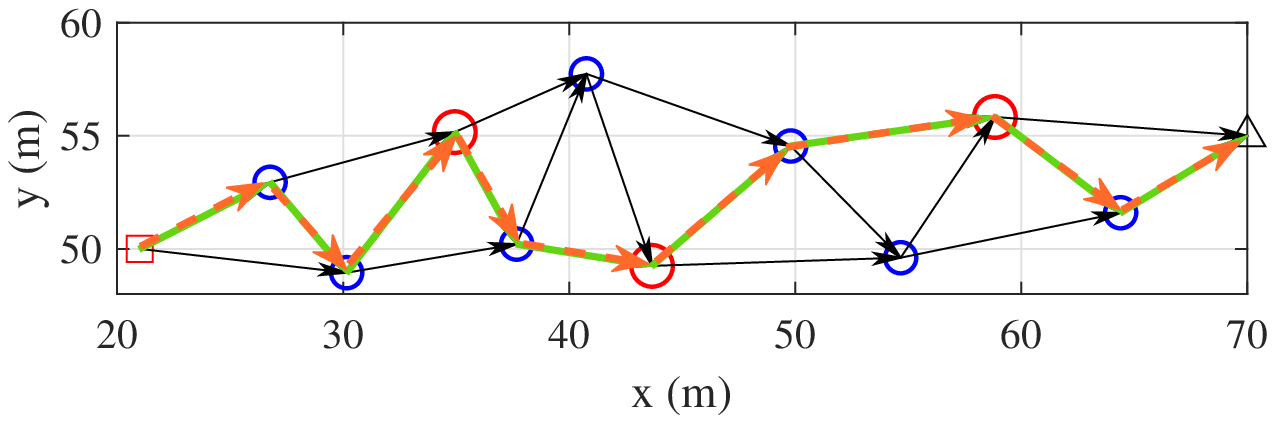}}
\subfigure[$N=100, M=1500$]{\label{Route1500_100}
	\includegraphics[width=8cm]{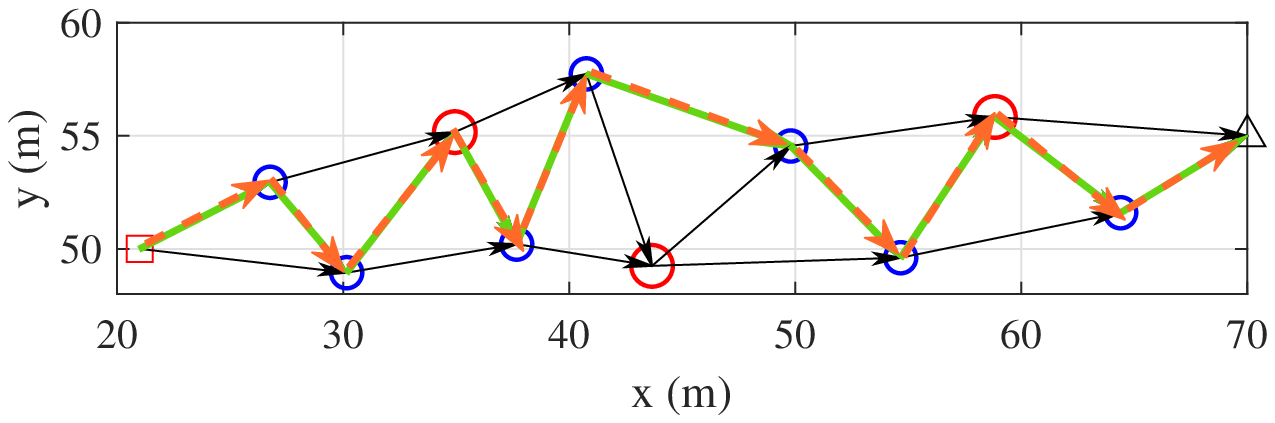}}
\subfigure[$N=1000, M=1500$]{\label{Route1500_1000}
	\includegraphics[width=8cm]{Route1500_100.eps}}
\caption{Different MAMP-IRS beam routing schemes versus $N$ and $M$.}\label{High_appro_route}
\vspace{-22pt}
\end{figure}
\vspace{-0.6cm}

\subsection{Effect of System Para meters on the Achievable Rate}
\vspace{-0.1cm}
Next, for rate performance comparison, we consider the following four benchmark schemes: 1) \emph{MAMP-IRS with optimal beam routing}, which traverses all feasible routing paths and selects the one that achieves the maximum receive SNR given in \eqref{SNR_out}; 2) \emph{MAMP-IRS with myopic beam routing}, which sequentially selects IRSs with the minimum edge weight for routing until reaching the user; 3) \emph{MAMP-IRS with random beam routing}, for which the IRSs are selected randomly; 4) \emph{all-passive-IRSs scheme}, for which the three active IRSs are replaced by three passive IRSs and the corresponding beam routing design is obtained by using the similar methods in \cite{mei2020cooperative}. 
\begin{figure}[t]
\centering
\subfigure[Achievable rate versus active-IRS amplification power.]{\label{MAMP_power}
\includegraphics[height=5.6cm]{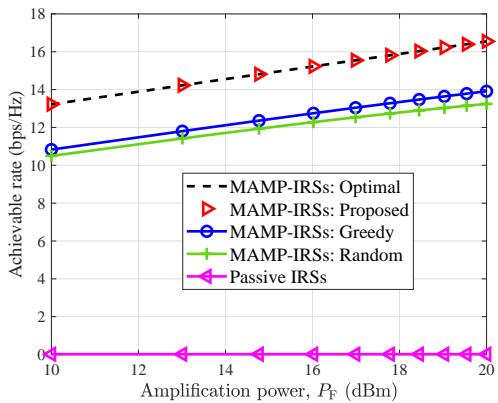}}
\hspace{10mm}
\subfigure[Achievable rate versus number of per-active-IRS reflecting elements.]{\label{MAMP_element}
\includegraphics[height=5.6cm]{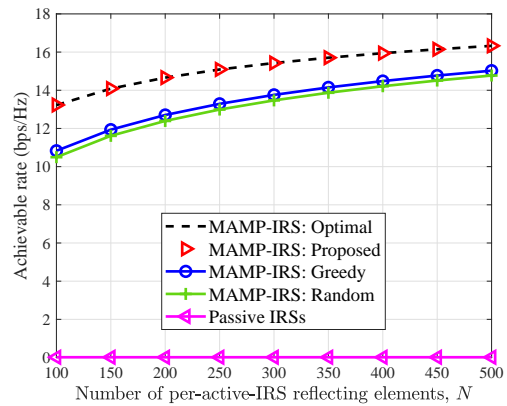}}
\subfigure[Achievable rate versus number of active IRSs.]{\label{Fig:NumActive}
    \includegraphics[height=5.6cm]{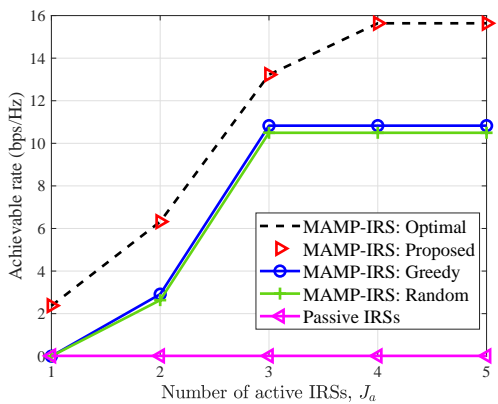}}
    \hspace{10mm}
\subfigure[Achievable rate versus Rician factor.]{\label{Rician}
    \includegraphics[height=5.6cm]{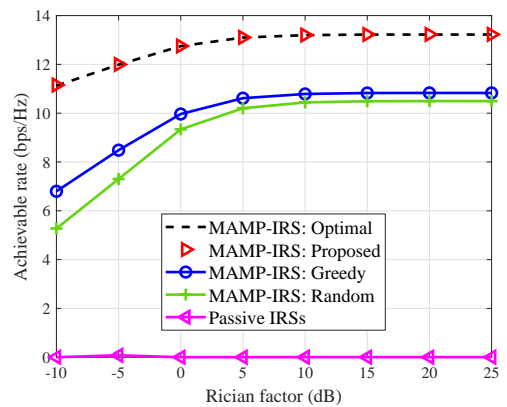}}
\caption{Effects of system parameters on the achievable rate.}\label{Fig:ParMAMP}
\vspace{-16pt}
\end{figure} 

\subsubsection{Effect of the Active-IRS Amplification Power}

We compare in Fig.~\ref{MAMP_power} the achievable rates of different MAMP-IRS beam routing schemes versus the active-IRS amplification power with $M=400$ and $N=100$. First, it is observed that the achievable rates of different MAMP-IRS beam routing schemes monotonically increase with the active-IRS amplification power. Second,
the proposed MAMP-IRS beam routing scheme achieves close rate performance to the optimal scheme and significantly outperforms the random and greedy benchmark schemes. Third, the achievable rate of the MAMP-IRS aided system is much higher than that of the all-passive-IRSs scheme at a modest higher energy cost, and the gap increases with the active-IRS amplification power. This is expected since a higher $P_{\rm F}$
tends to provide a higher power amplification gain.
\subsubsection{Effect of the Number of Per-Active-IRS Reflecting Elements}
Fig.~\ref{MAMP_element} plots the achievable rates of different MAMP-IRS beam routing schemes versus the number of reflecting elements at each active IRS, $N$. First, it is observed that the achievable rate of the MAMP-IRS aided system monotonically increases with the number of active reflecting elements $N$. Second,
the proposed MAMP-IRS beam routing scheme achieves near-optimal rate performance as compared to the optimal scheme and greatly outperforms the random and greedy benchmarks. {Last, the rate gap becomes larger as the per-active-IRS reflecting elements increases due to the higher beamforming gain.}
\subsubsection{Effect of the Number of Active IRSs}
{In Fig.~\ref{Fig:NumActive}, we show the effects of the number of active IRSs on the achievable rate  with $N=100$ and $M=400$. First, the achievable rate of the MAMP-IRS aided system first monotonically increases with the number of active IRSs, e.g., $J_{a}\leq4$, and then tends to saturate when $J_{a}$ is sufficiently large, e.g., $J_{a}>4$. This is because, with more active IRSs being involved in the beam routing path, the power amplification gain brought by the newly involved active IRS appears to be essentially negligible, resulting in a very slight rate increase. Second, the achievable rates of the random and greedy schemes keep unchanged when $J_{a} >3$, since the added active IRSs are not involved in their beam routing paths. Last, the rate gap between the MAMP-IRS aided system and passive IRSs aided system first increases with the number of active IRSs ($J_{a}\leq4$), and then becomes constant when $J_{a} >4$.}
\subsubsection{Effect of Rician Factor}
Finally, we consider the more practical Rician fading channel model for all involved MAMP-IRS reflection links, and plot in Fig.~\ref{Rician} the achievable rates of different schemes versus the Rician factor. Note that the joint beamforming of the BS, active IRSs, and passive IRSs is designed based on the dominant LoS channels, as shown in Lemma~\ref{Lem:MAMP_beamforming}. It is observed that with an increasing Rician factor, the achievable rates for all schemes first increase and then saturate when the Rician factor is sufficiently high. Moreover, even in the low-Rician-factor regime, the proposed MAMP-IRS beam routing scheme still achieves relatively high data rate and greatly outperforms  the other benchmark schemes.  

\vspace{-0.4cm}
\section{Conclusions}\label{CL}
\vspace{-0.2cm}
In this paper, we studied the multi-hop beam routing design for a new MAMP-IRS aided wireless system, where a number of active and passive IRSs cooperatively establish a virtual LoS multi-reflection path from the BS to the user. In particular, for the special case of the SAMP-IRS aided system, we showed that the active IRS should be selected to
establish the beam routing path when its amplification power
and/or number of active reflecting elements are sufficiently
large. Moreover, for the general MAMP-IRS aided wireless system, we proposed an efficient method to obtain a suboptimal beam routing design by using the hierarchical routing design method, receive SNR approximation, and graph theory. Numerical results demonstrated that
the proposed MAMP-IRS aided system attained substantial performance gains over various benchmark schemes under different system settings. In the future, it is an interesting direction to study efficient active/passive IRS placement strategy under the site/topology constraints to further improve the MAMP-IRS beam routing performance. Moreover, this work also opens up a new research direction for hybrid IRS aided wireless networks consisting of both active and passive IRSs.

\vspace{-0.2cm}
\begin{appendices}
\begin{spacing}{1.3}
\vspace{-0.3cm}
\section{}\label{App_2}
For any two adjacent active IRSs along any given active-IRS routing path $\Omega_{\rm act}$, denoted by $s_{\mu(i)},s_{\mu(i+1)} \in \Omega_{\rm act}$, the cascaded channel power gain in between is defined as $f_{s_{\mu(i)},s_{\mu(i+1)}}$.
Similar to Lemma~\ref{Lem:Decouple}, we first consider the effect of $f_{s_{\mu(i)},s_{\mu(i+1)}}$ on the receive SNR in problem (P$9$), the receive SNR is given by
\begin{equation}\label{Lem:General}
\!\!    \gamma^{(i)}_{\rm MAMP} \!\!= \!\!
\frac{P_{\rm B}N^{2K_{a}}f_{s_{\mu(K_{a})},J+1}\eta^{2}_{s_{\mu(K_{a})}}\!\!\left(\prod_{k_{a}=1}^{i-1}f_{s_{\mu(k_{a})}, s_{\mu(k_{a}+1})}\eta^{2}_{s_{\mu(k_{a})}}\!\prod_{k_{a}=i+1}^{K_{a}-1}f_{s_{\mu(k_{a})}, s_{\mu(k_{a}+1)}}\eta^{2}_{s_{\mu(k_{a})}}\right)\!\!f_{0, s_{\mu(1)}}}
{\frac{\sum_{k_{a}=1}^{K_{\mathrm{a}}}\left(f_{s_{\mu(K_{a})},J+1}N\eta^{2}_{s_{\mu(K_{a})}}\prod_{\ell=k_{a}}^{K_{a}-1}f_{s_{\mu(\ell)}, s_{\mu(\ell+1)}}\eta^{2}_{s_{\mu(\ell)}}N \sigma_{\mathrm{F}}^{2}\right)}{f_{s_{\mu(i)},s_{\mu(i+1)}}}+\frac{\sigma^2}{f_{s_{\mu(i)},s_{\mu(i+1)}}}}.
\vspace{-6pt}
\end{equation}
It can be verified from \eqref{Lem:General} that $\frac{\partial \gamma^{(i)}_{\mathrm{MAMP}}}{\partial f_{s_{\mu(i)},s_{\mu(i+1)}}} > 0$, and hence maximizing $ \gamma^{(i)}_{\rm MAMP}$ is equivalent to maximizing $f_{s_{\mu(i)},s_{\mu(i+1)}}$ and independent of $f_{s_{\mu(j)},s_{\mu(j+1)}}, \forall s_{\mu(j)},s_{\mu(j+1)}\in \Omega_{\rm act}, s_{\mu(j)} \neq s_{\mu(i)}$.
As such, one can easily conclude that, for any other adjacent active IRSs in $\Omega_{\rm act}$, similar result can be easily obtained as above. This leads to the result that maximizing the receive SNR over any given active-IRS routing path can be equivalently transformed to the channel power gain maximization of a set of sub-paths, thus completing the proof.
\end{spacing}
\vspace{-0.5cm}
\begin{spacing}{1.4}
\section{}\label{App_1}
{For ease of exposition, we denote $P_{s_{\mu(0)}} \triangleq P_{\rm B} $, $f^*_{s_{\mu(0)},s_{\mu(1)}} \triangleq f^*_{0,s_{\mu(1)}}$, $f^*_{s_{\mu(K_{a})},s_{\mu(K_{a}+1)}} \triangleq f^*_{s_{\mu(K_{a})},J+1}$, and $\sigma^2_{\mathrm{F}}\triangleq \sigma^2$ in the sequel.
Based on \eqref{MAMP_signal}, the received signal power via the MAMP-IRS multi-reflection path, $P_{\rm signal}$, is given by
\vspace{-0.2cm}
\begin{align}
    P_{\mathrm{signal}}=P_{s_{\mu(0)}}N^{2K_{a}}&(f^*_{s_{\mu(0)},s_{\mu(1)}} f^*_{s_{\mu(1)},s_{\mu(2)}}\cdots f^*_{s_{\mu(K_{a})},s_{\mu(K_{a}+1)}}) \times (\eta_{s_{\mu(1)}}^{2} \eta_{s_{\mu(2)}}^{2}\cdots \eta_{s_{\mu(K_{a})}}^{2}).
\end{align}
Accordingly, the superimposed amplification noise power at the receiver is given by
\vspace{-0.2cm}
\begin{align}
P_{\mathrm{noise}}&=\sigma^2_{\mathrm{F}}N^{K_{a}}(f^*_{s_{\mu(1)},s_{\mu(2)}} \cdots f^*_{s_{\mu(K_{a})},s_{\mu(K_{a}+1)}})(\eta_{s_{\mu(1)}}^{2} \eta_{s_{\mu(2)}}^{2} \cdots \eta_{s_{\mu(K_{a})}}^{2})\nn\\
    &+\sigma^2_{\mathrm{F}}N^{K_{a}-1}(f^*_{s_{\mu(2)},s_{\mu(3)}} \cdots f^*_{s_{\mu(K_{a})},s_{\mu(K_{a}+1)}})(\eta_{s_{\mu(2)}}^{2} \cdots \eta_{s_{\mu(K_{a})}}^{2})
    +\cdots+ \sigma^2_{\mathrm{F}}.
    \vspace{-6pt}
\end{align}
Consequently, the receive SNR is obtained as 
\begin{align}
    \gamma_{\mathrm{MAMP}}=\frac{ P_{\mathrm{signal}}}{P_{\mathrm{noise}}}=\frac{P_{s_{\mu(0)}}N^{2K_{a}}\prod_{k_{a}=1}^{K_{a}+1} f^*_{s_{\mu(k_{a}-1)},s_{\mu(k_{a})}} \prod_{k_{a}=1}^{K_{a}} \eta_{s_{\mu(k_{a})}}^{2}}{\sum_{k_{a}=1}^{K_{a}+1} \sigma^2_{\mathrm{F}} \prod_{t=k_{a}+1}^{K_{a}+1} f^*_{s_{\mu(t-1)},s_{\mu(t)}} \prod_{t=k_{a}}^{K_{a}} N\eta_{s_{\mu(t)}}^{2}}.\label{eq:rMAMP}
\end{align}
Dividing both the numerator and the denominator by $\prod_{k_{a}=1}^{K_{a}+1} \sigma^2_{\mathrm{F}}\prod_{k_{a}=1}^{K_{a}} N \eta_{s_{\mu(k_{a})}}^{2}$, the numerator in \eqref{eq:rMAMP}, denoted by $A$, is expressed as
\vspace{-0.3cm}
\begin{equation}
    A = P_{s_{\mu(0)}}N^{K_{a}}\prod_{k_{a}=1}^{K_{a}+1} \frac{f^*_{s_{\mu(k_{a}-1)},s_{\mu(k_{a})}}}{\sigma^2_{\mathrm{F}}}.
    \vspace{-6pt}
\end{equation}
Accordingly, the denominator, denoted by $B$, is given by
\begin{equation}\label{Denom}
    B = \sum_{k_{a}=1}^{K_{a}+1} \frac{ \prod_{t=k_{a}+1}^{K_{a}+1}f^*_{s_{\mu(t-1)},s_{\mu(t)}}/\sigma^2_{\mathrm{F}}}{\prod_{t=1}^{k_{a}-1} N\eta_{s_{\mu(t)}}^2 \prod_{t=1}^{k_{a}-1} \sigma_{\mathrm{F}}^2}.
\end{equation}
For each selected active IRS $s_{\mu(t)}\in\Omega_{\rm act}$, the amplification power constraint needs to be satisfied (i.e., \eqref{amplify}). 
Substituting~\eqref{amplify} into \eqref{Denom}, the denominator $B$ becomes
\begin{align}
    B =
     \sum_{k_{a}=1}^{K_{a}+1} \frac{ \prod_{t=k_{a}+1}^{K_{a}+1}f^*_{s_{\mu(t-1)},s_{\mu(t)}}/\sigma^2_{\mathrm{F}}}{\frac{\prod_{t=1}^{k_{a}-1} P_{s_{\mu(t)}}}{\prod_{t=1}^{k_{a}-1} P_{s_{\mu(t-1)}} f^*_{s_{\mu(t-1)}, s_{\mu(t)}}/\sigma_{\mathrm{F}}^{2}\prod_{i=t+1}^{k_{a}-1}N+\sum_{t=1}^{k_{a}-1}  \prod_{j=1}^{t-1} P_{s_{\mu(j-1)}} f^*_{s_{\mu(j-1)}, s_{\mu(j)}}/\sigma_{\mathrm{F}}^2\prod_{i=1}^{t-2}N }}.
\end{align}
Multiplying both the numerator $A$ and denominator $B$ by $\prod_{t=1}^{K_{a}} P_{s_{\mu(t)}}$, the numerator $A$ now becomes  
\begin{equation}
     \tilde{A} = N^{K_{a}}\prod_{k_{a}=1}^{K_{a}+1} \frac{P_{s_{\mu(k_{a}-1)}}f^*_{s_{\mu(k_{a}-1)},s_{\mu(k_{a})}}}{\sigma^2_{\mathrm{F}}},
     \vspace{-6pt}
\end{equation}
and the denominator $B$ is given by 
\begin{align}
    \!\!\!\tilde{B}\!\!=\!\!\sum_{k_{a}=1}^{K_{a}+1} \!\! \left[\prod_{t=k_{a}+1}^{K_{a}+1} \!\!\frac{ P_{s_{\mu(t-1)}}f^*_{s_{\mu(t-1)},s_{\mu(t)}}}{\sigma^2_{\mathrm{F}} }
   \!\!\left(
    \prod_{t=1}^{k_{a}-1}\!\! \frac{P_{s_{\mu(t-1)}} f^*_{s_{\mu(t-1)}, s_{\mu(t)}}}{\sigma_{\mathrm{F}}^2}\! \! \prod_{i=t+1}^{k_{a}-1}\! \! N
    \!\!+\!\!\sum_{t=1}^{k_{a}-1}  \prod_{j=1}^{t-1} \frac{P_{s_{\mu(j-1)}} f^*_{s_{\mu(j-1)}, s_{\mu(j)}}}{\sigma_{\mathrm{F}}^2}\prod_{i=1}^{t-2}N \!\!\right)
    \right].\nn
\end{align}
Let $\gamma_{i} \triangleq P_{s_{\mu(i-1)}} f^*_{s_{\mu(i-1)},s_{\mu(i)}}/ \sigma^2_{\mathrm{F}}$ denote the per-hop SNR between nodes $s_{\mu(i-1)}$ and $s_{\mu(i)}$, based on which the numerator $\tilde{A}$ and the denominator $\tilde{B}$ can be respectively rewritten as 
\begin{equation}
    \tilde{A} =  N^{K_{a}}\prod_{k_{a}=1}^{K_{a}+1} \gamma_{k_{a}},~~~
    \tilde{B} =
    \sum_{k_{a}=1}^{K_{a}+1}  \prod_{t=k_{a}+1}^{K_{a}+1}\gamma_{t}
    \left(
    \prod_{t=1}^{k_{a}-1}\gamma_{t}\prod_{i=t+1}^{k_{a}-1}N +\sum_{t=1}^{k_{a}-1}  \prod_{j=1}^{t-1} \gamma_{j}\prod_{i=1}^{t-2}N\right).
\end{equation}
The equivalent SNR expression in MAMP-IRS aided system is then given by
\begin{align}
  \gamma_{\mathrm{MAMP}}&=\frac{N^{K_{a}}\prod_{k_{a}=1}^{K_{a}+1} \gamma_{k_{a}}}{\sum_{k_{a}=1}^{K_{a}+1}  \prod_{t=k_{a}+1}^{K_{a}+1}\gamma_{t}
    (
    \prod_{t=1}^{k_{a}-1}\gamma_{t}\prod_{i=t+1}^{k_{a}-1}N +\sum_{t=1}^{k_{a}-1}  \prod_{j=1}^{t-1} \gamma_{j}\prod_{i=1}^{t-2}N)}\nn\\
&\stackrel{(b_1)}{=}\!\frac{N^{K_{a}}}{\sum_{k_{a}=1}^{K_{a}+1} \sum_{t=1}^{k_{a}} \prod_{i=t}^{k_{a}} \frac{1}{ \gamma_{i}}\prod_{j=1}^{t-2}N} ,
\end{align}
where $(b_1)$ is obtained by dividing both the numerator and the denominator by $\prod_{k_{a}=1}^{K_{a}+1} \gamma_{k_{a}}$.\\
Combining the above steps leads to the desired result.}
\end{spacing}
\end{appendices}

\vspace{-0.2cm}
\begin{spacing}{1.15}

\end{spacing}

\end{document}